\documentclass[11pt,]{amsart}
\usepackage{multicol}
\usepackage[dvipsnames]{xcolor}
\usepackage[utf8]{inputenc}
\usepackage[a4paper, margin=2cm]{geometry}
\usepackage{amsmath}
\usepackage{amsfonts}
\usepackage{amssymb}
\usepackage{amsthm}
\usepackage{pgfplots}
\usepackage{graphicx}
\usepackage{lmodern}
\usepackage{cancel}
\usepackage{physics}
\usepackage{mathtools}
\usepackage{dsfont}
\usepackage{tikz}
\usetikzlibrary{calc}
\usepackage{fullpage}
\usepackage{enumitem}
\usetikzlibrary{positioning}
\usepackage{capt-of}
\usepackage[colorlinks]{hyperref}
\hypersetup{
    colorlinks  = true,
    citecolor   = PineGreen,
    linkcolor   = MidnightBlue,
    urlcolor    = PineGreen
}

\usepackage{thmtools}
\declaretheorem[parent=section]{theorem}
\declaretheorem[numberlike=theorem, name=Proposition]{proposition}

\declaretheorem[numberlike=theorem, name=Example]{example}
\declaretheorem[numberlike=theorem, name=Lemma]{lemma}
\declaretheorem[numberlike=theorem, name=Definition]{definition}
\declaretheorem[numberlike=theorem, name=Corollary]{corollary}
\declaretheorem[numberlike=theorem, name=Question]{question}

\def\ket#1{|#1\rangle}

\usepackage{fancyhdr} 
\usepackage[foot]{amsaddr}

\title{\large On approximate quantum error correction for symmetric noise}

\author{Gereon Koßmann$^{1,*}$}
\address{$^1$ Institute for Quantum Information, RWTH Aachen University, Aachen, Germany}
\author{Julius A. Zeiss$^{1}$}
\author{Omar Fawzi$^{2}$}
\address{$^2$ Inria, ENS de Lyon, UCBL, LIP, 69342, Lyon Cedex 07, France}
\author{Mario Berta$^{1,3}$}
\address{$^3$ Department of Computing, Imperial College London, London, UK}
\address{$^*$ Corresponding author: kossmann@physik.rwth-aachen.de}

\begin{document}

\begin{abstract}
We revisit the extendability-based semi-definite programming hierarchy introduced by Berta {\it et al.}~[Mathematical Programming, 1-49 (2021)], which provides converging outer bounds on the optimal fidelity of approximate quantum error correction (AQEC). As our first contribution, we introduce a measurement-based rounding scheme that extracts inner sequences of certifiably good encoder-decoder pairs from this outer hierarchy. To address the computational complexity of evaluating fixed levels of the hierarchy, we investigate the use of symmetry-based dimension reduction. In particular, we combine noise symmetries\,---\,such as those present in multiple copies of the qubit depolarizing channel\,---\,with the permutational symmetry arising from the extendability of the optimization variable. This framework is illustrated through basic, but already challenging numerical examples that showcase its practical effectiveness. Our results contribute to narrowing the gap between theoretical developments in quantum information theory and their practical applications in the analysis of small-scale quantum error-correcting codes.
\end{abstract}

\maketitle


\section{Introduction}

\subsection{Motivation}

Quantum computing is increasingly advancing from theoretical exploration to practical implementation, with actual computations on quantum devices becoming more and more feasible. A core requirement for achieving reliable results is addressing the fundamental challenge of quantum error correction. Generally, the error correction problem seeks an encoder-decoder pair, denoted by $(\mathcal{E},\mathcal{D})$, that enables the faithful recovery of quantum information under a noisy channel $\mathcal{N}$. Recent years have seen initial experimental demonstrations and implementations of quantum error-correcting codes (see, e.g., \cite{Krinner_2022,acharya2024quantumerrorcorrectionsurface} and references therein).

Conceptually, we distinguish between zero-error codes as usually considered for quantum fault-tolerance and approximate quantum error correction (AQEC), which tolerates a small amount of error \cite{schumacher2001approximatequantumerrorcorrection, B_ny_2010}. AQEC is motivated by practical considerations in physical hardware, such as the cost of qubits and the varying frequencies of different error types. It can also be viewed from the perspective of quantum Shannon theory, where it is understood that allowing a small amount of error is conceptually strikingly different from zero-error information theory. For example, finite errors become a powerful tool for addressing finite-size rate considerations for the quantum channel capacity \cite{Tomamichel_2016}.

AQEC approaches treat noisy channels by solving an optimization problem, finding the threshold between exact error correction and the best achievable approximate value within a given model, as defined by a selected metric. Reimpell and Werner \cite{Reimpell_2005} pioneered a widely recognized method for channel fidelity which addresses this optimization using see-saw algorithms over encoder-decoder pairs, providing inner bounds (see also \cite{PhysRevA.75.012338,Kosut:2009aa,5429112,johnson2017qvectoralgorithmdevicetailoredquantum}). However, there is the challenge of local minima in see-saw approaches for non-convex problems, often causing the algorithm to converge to a flat region. In particular, there is no certificate for optimality. Nevertheless, the see-saw approach is particularly effective for a first educated guess for an encoder-decoder pair, because it iteratively optimizes over product states\,---\,the Choi matrices of the encoder and decoder\,---\,in a simple stepwise manner.

Mathematically, the AQEC problem corresponds to a bilinear optimization in non-commutative variables of fixed dimension and with the semi-definite Choi constraints from the encoder-decoder pair. Consequently, AQEC can be understood as an instance of a constrained quantum separability (cSEP) problem \cite{berta2021semidefinite}, with the term cSEP borrowed from \cite{ohst2024characterisingmemoryquantumchannel}. As a well-known tool, the set of separable quantum states (SEP) can be approximated from the outside in terms of $n$-extendable states, i.e., states with a permutation-invariant $n$-extension. This leads to the Doherty-Parrilo-Spedalieri (DPS) semi-definite programming (SDP) hierarchy \cite{doherty2004complete}. So, even though determining membership in the set of separable quantum states is NP-hard \cite{doherty2004complete,gurvits2003classical}, this framework then allows optimization over a fixed number of copies in the symmetric subspace, thereby yielding provable outer bounds that can be computed efficiently for every level of the hierarchy. Adapting the DPS framework to AQEC, however, introduces additional challenges due to the linear constraints mandated by Choi's theorem. In fact, since the DPS hierarchy relies on quantum de-Finetti theorems \cite{caves2002unknown} without linear constraints, it cannot address the requirements of AQEC or general cSEP problems \cite[Exam.~4.6]{berta2021semidefinite}. To overcome this limitation, novel quantum de-Finetti theorems were developed in \cite{berta2021semidefinite,brandao_harrow}, allowing the construction of SDP hierarchies incorporating linear constraints (also see our concurrent work~\cite{zeiss25}).


\subsection{Overview of results}

This work addresses several challenges associated with the SDP hierarchy from \cite{berta2021semidefinite}. First, proposes the development of a rounding process to bridge the gap between outer approximations and the inner see-saw methods from \cite{Reimpell_2005}. Second, it explores how to leverage symmetries of the noisy channel of the error model for dimension reduction. More specifically, there are the following contributions:\\
 
\begin{itemize}
    \item \textbf{Rounding outer solutions.} One natural way to combine the inner bounds from \cite{Reimpell_2005} with the outer bounds from \cite{berta2021semidefinite} is through their objective values. Namely, if the value produced by the see-saw method is close to that of the outer bound, this can serve as a form of optimality certification. However, the bilinear structure of the underlying optimization problem raises the question of whether outer solutions can be rounded to yield feasible inner points, in analogy to the Goemans-Williamson algorithm \cite{10.1145/227683.227684}, which rounds solutions from outer approximations of quadratic unconstrained binary optimization problems. We develop a quantum-inspired rounding technique based on measurements, which yields a sequence of AQEC codes with rigorous optimality guarantees (\autoref{thm:inner_bounds}). These AQEC codes may then, for example, serve as effective warm starts for see-saw optimization.\\
    
    \item \textbf{Combined symmetry reduction for AQEC.} The \emph{extendibility symmetry} of $n$-extendable states can be leveraged using representation-theoretic dimension reduction techniques from semi-definite programming \cite{Klerk2007ReductionOS}. For the problem of AQEC featuring the Choi constraints, this was carefully worked out by Chee {\it et al.}~\cite{chee2024efficient}. However, in practical scenarios one is often interested in the problem of encoding a single logical qubit into multiple identical and independently distributed (iid) noisy physical qubits, for instance, subject to single-qubit depolarizing noise. These noisy qubits naturally exhibit permutation invariance from the iid structure, which we refer to as \emph{iid-symmetry}. Additionally, individual qubits may possess covariance under the action of the single-qubit unitary group\,---\,as is the case with depolarizing channels. In evaluating fixed levels of the corresponding SDP hierarchies, it is then desirable to simultaneously exploit these different symmetries. We introduce a framework based on \emph{commuting group actions} to handle symmetries arising in the context of AQEC. By leveraging the representation theory of the symmetric group \cite{Sagan2001}, we enhance the feasibility of efficiently computing outer bounds for AQEC (\autoref{thm:small_SDPs}). In particular, we provide explicit examples involving the depolarizing channel to illustrate how the framework enables the practical computation of non-trivial outer bounds (\autoref{subsec:commutative_relax_examples})\,---\,specifically, those requiring at least one extension and a decent number of physical qubit. We note that such cases were previously intractable due to the overwhelming size of the associated optimization problems (\autoref{sec:numerical_results}).\\

    \item \textbf{Combining symmetries of optimization variable and objective function.} The AQEC problem opens the door for a general discussion of combined symmetry reduction in semi-definite programs. In \autoref{sec:framework}, we are concerned with the general question of having a semi-definite program where the objective function and the constraints are invariant with respect to one symmetry group, while the optimization variable is invariant under a second, generally unrelated group. We extend the symmetry reduction framework introduced in \cite{Klerk2007ReductionOS} to the more general situation of combined symmetry reduction and provide necessary and sufficient criteria for solving the SDP with respect to a new group built from the groups corresponding to the objective function and the optimization variable. We use the AQEC problem to demonstrate that it is indeed crucial to carefully adjust the action of the symmetry groups under consideration so that they satisfy the condition provided in this section. To illustrate this, we provide a concrete example in the language of AQEC where a misleading choice of group actions leads to an intractable problem for joint symmetry reduction. As a byproduct, we revisit our finding from \autoref{sec:commutative_relaxation} as the generic case of commuting group actions and thus combined symmetry reduction.\\
\end{itemize}

The remainder of this manuscript is organized as follows. In \autoref{sec:approx_quantum_error}, we revisit the AQEC framework, recall the reformulated bilinear optimization problem, and present a method for simplifying the convex analysis of AQEC. In \autoref{sec:inner_bounds}, we construct a rounding algorithm that transforms outer solutions into feasible inner points, with a corresponding convergence analysis. In \autoref{sec:symmetry_reduction}, we introduce the symmetry reduction framework and discuss the general use of positive, but not completely positive, maps as separability criteria. We also motivate why combining symmetries can much simplify actual computations at higher levels of the SDP hierarchy from \cite{berta2021semidefinite}. \autoref{sec:commutative_relaxation} introduces a functional approach to symmetry reduction for AQEC, yielding convergent outer bounds with reduced dimensionality. In \autoref{sec:numerical_results}, we discuss the practical implementation of our methods and present proof-of-principle numerical results. Finally, in \autoref{sec:framework} we generalize the arguments for combining symmetries in SDPs to arbitrary compact subgroups of the unitary group $\mathcal{U}(\mathcal{H})$. We emphasize that \autoref{sec:inner_bounds} can similarly be found for fixed-size, two player free non-local games in our concurrent work~\cite{zeiss25}.


\section{Approximate quantum error correction (AQEC)}
\label{sec:approx_quantum_error}

\subsection{Framework}

This section introduces the AQEC framework as in \cite{berta2021semidefinite}. We consider a general noisy channel $\mathcal{N}$ as the source of errors and seek optimal encoder-decoder pairs $(\mathcal{E},\mathcal{D})$ that approximately preserve a logical subspace, rendering it resilient to noise. Our goal is to optimize over encoding and decoding operations that mitigate the effects of the noise introduced by the channel, ensuring reliable information transmission. This approach allows us to frame error correction as an optimization problem, where the objective is to find the best possible encoding and decoding strategy to maximize error resilience. Thus, in abstract terms, the error correction process can be described by the following sequence of channel operations:
\begin{equation}\label{eq:setting_error_correction}
\begin{aligned}
    \mathcal{S}(\mathcal{H}_L) \stackrel{\mathcal{E}}{\longrightarrow} \mathcal{S}(\mathcal{H}_P) \stackrel{\mathcal{N}}{\longrightarrow} \mathcal{S}(\mathcal{H}_P) \stackrel{\mathcal{D}}{\longrightarrow} \mathcal{S}(\mathcal{H}_L).
\end{aligned}
\end{equation}
In \eqref{eq:setting_error_correction}, we consider a noisy channel $\mathcal{N}$ acting on a physical system $\mathcal{H}_P$ and aim to identify an encoder channel $\mathcal{E}$ and a decoder channel $\mathcal{D}$ that effectively mitigate the impact of noise introduced by $\mathcal{N}$. Our objective is to design these channels so that the encoded information is approximately robust against general noise, ensuring the original logical information is accurately reconstructed after transmission. By approximating a noiseless transmission through this sequence of channels, we can safeguard quantum information despite the presence of errors, as defined more formally next.
Given an encoder-decoder pair $(\mathcal{E},\mathcal{D})$, we require a suitable metric to evaluate the proximity of the composed channel $\mathcal{D} \circ \mathcal{N} \circ \mathcal{E}$ to the identity on $\mathcal{S}(\mathcal{H}_L)$:
\begin{align}\label{eq:equation_approximate_equality}
    \mathcal{D} \circ \mathcal{N} \circ \mathcal{E} \approx \operatorname{id}_L.
\end{align}
Two common metrics are often considered to evaluate the effectiveness of error correction: the diamond norm and the channel fidelity. The diamond norm serves as a one-shot measure in channel discrimination tasks and can be seen as a worst-case bound for \eqref{eq:equation_approximate_equality}. In contrast, channel fidelity measures the accuracy of entanglement preservation and transmission, providing insight into how well average quantum information can be transmitted. Formally, the channel fidelity is defined as
\begin{equation}\label{eq:definition_channel_fidelity}
\begin{aligned}
    F\left(\mathcal{N}, d_L \right)\coloneqq  &\sup_{(\mathcal{E},\mathcal{D})} \mathcal{F} \left[\Phi_{L\bar{L}},(\mathcal{D}\circ \mathcal{N}\circ\mathcal{E})\otimes \operatorname{id}_L(\Phi_{L\bar{L}})\right] \\
    \operatorname{s.th.}& \quad \mathcal{E}:\mathcal{S}(\mathcal{H}_L) \to \mathcal{S}(\mathcal{H}_P) \quad \text{cptp}, \\ 
    &\quad \mathcal{D}:\mathcal{S}(\mathcal{H}_P) \to \mathcal{S}(\mathcal{H}_L) \quad \text{cptp}.
\end{aligned}
\end{equation}
Here, $\Phi_{L\bar{L}}$ represents a (normalized) maximally entangled state on the logical Hilbert space $\mathcal{H}_L$, serving as a reference for measuring the preservation of entanglement through the channel. The expression for channel fidelity can be reformulated into a bilinear optimization problem involving both the encoder and decoder channels \cite{Reimpell_2005} as
\begin{equation}\label{eq:resulting_opt}
\begin{aligned}
       F(\mathcal{N},d_L) \ = \   d_P^2  \max& \operatorname{tr}\left[ C_{\operatorname{id}_L \otimes \mathcal{N}}(C_{\mathcal{D}^\dagger}^T \otimes C_\mathcal{E})\right]& \\ 
        \operatorname{s.th.}& \ C_\mathcal{E} \geq, 0 \quad C_{\mathcal{D}^\dagger} \geq 0& \\ 
        & \ \operatorname{tr}_{\bar{P}}[C_\mathcal{E}] = \mathds{1}_{\bar{L}}/d_{\bar{L}}, \quad \operatorname{tr}_L[C_{\mathcal{D}^\dagger}] = \mathds{1}_P/d_P,&
\end{aligned}
\end{equation}
where $C_{-\cdot -}$ denotes the unique normalized Choi matrix corresponding to the given channel.


\subsection{A tractable outer optimization program}

\subsubsection{Convexification and hardness of outer approximations}

Due to the linearity of the objective function in \eqref{eq:resulting_opt} we can optimize over all convex mixtures of product states satisfying the constraints:
\begin{align}\label{eq:def_C}
    &\Sigma_{\operatorname{prod}}(LP:\bar{L}\bar{P})\nonumber\\
    &\coloneqq  \operatorname{conv} \left\{\rho_{LP} \otimes \sigma_{\bar{L}\bar{P}} \ \middle| \ \rho_{LP},\sigma_{\bar{L}\bar{P}} \in \mathcal{S}(\mathcal{H}_{LP,\bar{L}\bar{P}}), \ \operatorname{tr}_{\bar{P}}[\sigma_{\bar{L}\bar{P}}] = \frac{\mathds{1}_{\bar{L}}}{d_{\bar{L}}}, \ \operatorname{tr}_{L}[\rho_{LP}] = \frac{\mathds{1}_{P}}{d_P} \right\}.
\end{align}
With this definition, the problem of optimizing the channel fidelity in \eqref{eq:resulting_opt} is reformulated as
\begin{equation}\label{eq:optimization_C}
\begin{aligned}
       F(\mathcal{N},d_L) \ = \ d_P^2 \max& \operatorname{tr}\left[ C_{\operatorname{id}_L \otimes \mathcal{N}} \rho_{LP\bar{L}\bar{P}}\right]& \\ 
        \operatorname{s.th.}& \ \rho_{LP\bar{L}\bar{P}} \in \Sigma_{\operatorname{prod}}(LP:\bar{L}\bar{P}),&
\end{aligned}
\end{equation}
and can be understood as an instance of a constrained separability (cSEP) problem \cite{ohst2024characterisingmemoryquantumchannel}. For the full set of separable states (SEP)
\begin{align}\label{eq:def_sep_states}
    &\operatorname{SEP}(A:B)\nonumber\\
    & \coloneqq\Big\{\rho_{AB} \in \mathcal{S}(\mathcal{H}_A \otimes \mathcal{H}_B) \ \Big\vert \ \rho_{AB} = \sum_{i} p_i \rho_A^i \otimes \rho_B^i, \ p_i\geq 0, \  \sum_i p_i = 1, \ \rho_{A}^i,\rho_B^i \in \mathcal{S}(\mathcal{H}_{A,B})\Big\},
\end{align}
one then has that $\Sigma_{\operatorname{prod}}(LP:\bar{L}\bar{P}) \subset \operatorname{SEP}(LP:\bar{L}\bar{P})$. Unfortunately, the results from \cite{fawzi2019setseparablestatesfinite} for $\operatorname{SEP}(A:B)$ already imply that $\Sigma_{\operatorname{prod}}(LP:\bar{L}\bar{P})$ lacks an exact semi-definite representation in case that the $P$-system and the $\bar{L}$-system are trivial, which also aligns with the hardness results established in \cite{doherty2004complete,gurvits2003classical}. As such, conic hierarchies can serve as a valuable tool for approximating convex sets that cannot be expressed as semi-definite programs\,---\,prominent examples including \cite{doherty2004complete,lasserre2001global}.


\subsubsection{An SDP hierarchy}

For the convex set $\Sigma_{\operatorname{prod}}(LP:\bar{L}\bar{P})$, the work \cite{berta2021semidefinite} provided an outer SDP hierarchy, which can abstractly be written as a sequence of nested sets
\begin{align}\label{eq:def_sdp_hierarchies}
    \Sigma_{\operatorname{prod}}(LP:\bar{L}\bar{P}) \subset \cdots \subset \Sigma^n_{\operatorname{prod}}(LP:\bar{L}\bar{P}) \subset \Sigma^{n-1}_{\operatorname{prod}}(LP:\bar{L}\bar{P})\subset \cdots \subset \Sigma^{1}_{\operatorname{prod}}(LP:\bar{L}\bar{P}),
\end{align}
with $\rho_{LP\bar{L}\bar{P}} \in \Sigma_{\operatorname{prod}}^n(LP:\bar{L}\bar{P})$ if there exists an extension $\rho_{LP(\bar{L}\bar{P})^{(1\ldots n)}}$ on $n$ copies of $\bar{L}\bar{P}$ such that
\begin{enumerate}
        \item[(i)] $\rho_{LP(\bar{L}\bar{P})^{(1\ldots n)}}\succeq 0$ and $\operatorname{tr}[\rho_{LP(\bar{L}\bar{P})^{(1\ldots n)}}]=1$
        \vspace{1mm}
        \item[(ii)] $\rho_{LP(\bar{L}\bar{P})^{(1\ldots n)}} = U^{\sigma}_{(\bar{L}\bar{P})^{1\ldots n}}\left(\rho_{LP(\bar{L}\bar{P})^{(1\ldots n)}}\right) \quad \forall\sigma\in S_n$
        \vspace{1mm}
        \item[(iii)] $\operatorname{tr}_{L}\left[ \rho_{LP(\bar{L}\bar{P})^{(1\ldots n)}} \right] = \frac{\mathds{1}_P}{d_P} \otimes\rho_{(\bar{L}\bar{P})^{(1\ldots n)}}$
        \vspace{1mm}
        \item[(iv)] $\operatorname{tr}_{\bar{P}^{(n)}}\left[\rho_{(\bar{L}\bar{P})^{(1\ldots n)}}\right] = \rho_{(\bar{L}\bar{P})^{(1\ldots n-1)}}\otimes\frac{\mathds{1}_{\bar{L}^{(n)}}}{d_{\bar{L}^{(n)}}}$,
\end{enumerate}
where $S_n$ denotes the symmetric group of order $n$ and $U^{\sigma}_{(\bar{L}\bar{P})^{1\ldots n}}$ its natural actions on $(\bar{L}\bar{P})^{(1\ldots n)}$. Consequently, the $n$-level outer approximation is formulated as
\begin{equation}\label{eq:nlevel_relax}
\begin{aligned}
       \operatorname{SDP}_n^\star \coloneqq  d^2_P  \max& \operatorname{tr}\left[ C_{(\operatorname{id}_L \otimes \mathcal{N})} \rho_{LP\bar{L}\bar{P}}\right]& \\ 
        \operatorname{s.th.}& \ \rho_{LP\bar{L}\bar{P}} \in \Sigma^n_{\operatorname{prod}}(LP:\bar{L}\bar{P}),&
\end{aligned}
\end{equation}
which represents an adapted SDP hierarchy for the cSEP problem of AQEC.


\subsection{Types of symmetries}\label{subsec:types_of_symmetries}

In \eqref{eq:nlevel_relax}, the following symmetries are typically present, where we distinguish between symmetries arising from the objective function—i.e., the Choi matrix in AQEC—and the extendibility symmetry, which inherently arises from the optimization variable:
\begin{enumerate}
    
    \item[(a)]\label{extendibility symmetry}\textbf{Symmetry of the optimization variable:}: An \emph{extendibility symmetry} $S_n$ over the copies of the $\bar{L}\bar{P}$-system arising via
    \begin{align*}
        \rho_{LP\bar{L}\bar{P}} = \operatorname{tr}_{(\bar{L}\bar{P})^{(2\ldots n)}}[\rho_{LP(\bar{L}\bar{P})^{(1\ldots n)}}]\quad\rho_{LP(\bar{L}\bar{P})^{(1\ldots n)}} = U^{\sigma}_{(\bar{L}\bar{P})^{1\ldots n}}\left(\rho_{LP(\bar{L}\bar{P})^{(1\ldots n)}}\right) \quad \forall \sigma \in S_n,
    \end{align*}
    whereby $U^{\sigma}_{(\bar{L}\bar{P})^{1\ldots n}}\left( \cdot \right)$ acts as the permutation of tensor factors.
    \vspace{1mm}
    \item[(b)]\label{isotropic state symmetry} \textbf{Symmetry of the objective function:} A symmetry arises from the observation that there exists an isomorphism of Hilbert spaces such that
    \begin{align*}
        C_{\operatorname{id}_L \otimes \mathcal{N}} \mapsto C_{\operatorname{id}_L} \otimes C_{\mathcal{N}}.
    \end{align*}
    This means that the action of the identity channel on the logical space $ \mathcal{H}_L $ combined with the noisy channel $ \mathcal{N} $ can be represented as the tensor product of their respective Choi matrices. The maximally entangled state associated with this formulation is an isotropic state \cite{Werner1989}, which is invariant under unitary $2$-designs, expressed as
    \begin{align*}
        (U \otimes \overline{U}) C_{\operatorname{id}_L}  \left(U\otimes \overline{U}\right)^\dagger = C_{\operatorname{id}_L}, \quad U \in \mathcal{U}(\mathcal{H}_L).
    \end{align*}
    We call this the \emph{isotropic state symmetry}.
    \vspace{1mm}
    \item[(c)]\label{iid symmetry} \textbf{Symmetry of the objective function:}
    When considering multiple qubits it can be reasonable to assume that each qubit experiences similar errors. This leads to the perspective of permutation-invariant noise in $ \mathcal{N} $, or even a model where the noise is identical and independent (iid) across qubits. If we assume that $ \mathcal{N} \equiv \mathcal{M}^{\otimes m} $, for a local channel $ \mathcal{M}_i:\mathcal{S}(\mathcal{H}_{P_i}) \to \mathcal{S}(\mathcal{H}_{P_i})$, this gives rise to the following isomorphism of Hilbert spaces
    \begin{align*}
        C_{\mathcal{N}} \mapsto C_{\mathcal{M}}^{\otimes m}.
    \end{align*}
    We call this the \emph{iid-symmetry}.
    \vspace{1mm}
    \item[(d)]\label{internal Choi symmetry} \textbf{Symmetry of the objective function:}
    Let us assume that the local channel $ \mathcal{M} $ in (c) is a depolarizing channel acting on a qubit, which introduces an additional layer of internal symmetry in $\mathcal{M}$ \cite[Appendix A]{berta2021semidefinite}. This channel is covariant under the unitary group (see \autoref{subsec:symmetries_of_choi}) and thus has this additional \emph{Choi symmetry}. 
\end{enumerate}

These observations about the potential local symmetries in $ C_{\operatorname{id}_L \otimes \mathcal{N}} $ raise the question of how such symmetries can influence the semi-definite representation of $ \Sigma_{\operatorname{prod}}^n(LP:\bar{L}\bar{P})$. To address symmetry reduction, general semi-definite programming techniques have been developed, as discussed in \cite{Klerk2007ReductionOS}.


\subsection{Separability versus constrained separability}
\label{subsec:definition_of_hierarchies}

Here we aim to work out the proper difference between the hierarchy \cite{berta2021semidefinite}\,---\,which converges to \eqref{eq:optimization_C}\,---\,and a simplified hierarchy, which is directly implementable in the standard DPS framework. For this purpose we juxtapose next to $\Sigma_{\operatorname{prod}}(LP:\bar{L}\bar{P})$, defined in \eqref{eq:def_C}, the set
\begin{equation}
\begin{aligned}
&\Sigma_{\operatorname{SEP}}(LP:\bar{L}\bar{P})\coloneqq  \left\{\rho_{LP\bar{L}\bar{P}} \in \operatorname{SEP}(LP:\bar{L}\bar{P}) \ \middle\vert \ \operatorname{tr}_P[\rho_{LP}] = \frac{\mathds{1}_L}{d_L} \quad \text{and} \quad \operatorname{tr}_{\bar{L}}[\rho_{\bar{L}\bar{P}}] = \frac{\mathds{1}_{\bar{P}}}{d_{\bar{P}}}\right\}\label{eq:def_D}
\end{aligned}
\end{equation}
and further define outer approximations for both of them in terms of
\begin{align}\label{eq:comparison_hierarchies}
    &\begin{array}{ll}
        \rho_{LP\bar{L}\bar{P}} \in \Sigma_{\operatorname{prod}}^n(LP:\bar{L}\bar{P}) \Longleftrightarrow \\
        \\
        \text{(i)} \quad \rho_{LP(\bar{L}\bar{P})^{(1\ldots n)}}\succeq 0, \quad \operatorname{tr}[\rho_{LP(\bar{L}\bar{P})^{(1\ldots n)}}]=1 \\[1ex]
        \vspace{1mm}
        \text{(ii)} \quad \rho_{LP(\bar{L}\bar{P})^{(1\ldots n)}} = U^{\sigma}_{(\bar{L}\bar{P})^{1\ldots n}}\left(\rho_{LP(\bar{L}\bar{P})^{(1\ldots n)}}\right) \\[1ex]
        \vspace{1mm}
        \text{(iii)} \quad \operatorname{tr}_{L}\left[\rho_{LP(\bar{L}\bar{P})^{(1\ldots n)}} \right] = \frac{\mathds{1}_P}{d_P} \otimes\rho_{(\bar{L}\bar{P})^{(1\ldots n)}} \\[1ex]
        \vspace{1mm}
        \text{(iv)} \quad \operatorname{tr}_{\bar{P}^{(n)}}\left[\rho_{(\bar{L}\bar{P})^{(1\ldots n)}}\right] = \rho_{(\bar{L}\bar{P})^{(1\ldots n-1)}}\otimes\frac{\mathds{1}_{\bar{L}^{(n)}}}{d_{\bar{L}^{(n)}}}
    \end{array}
    &
    \begin{array}{ll}
    \rho_{LP\bar{L}\bar{P}} \in \Sigma_{\operatorname{SEP}}^n(LP:\bar{L}\bar{P}) \Longleftrightarrow \\
        \\
        \text{(i)} \quad \rho_{LP(\bar{L}\bar{P})^{(1\ldots n)}}\succeq 0, \quad \operatorname{tr}[\rho_{LP(\bar{L}\bar{P})^{(1\ldots n)}}]=1 \\[1ex]
        \vspace{1mm}
        \text{(ii)} \quad \rho_{LP(\bar{L}\bar{P})^{(1\ldots n)}} = U^{\sigma}_{(\bar{L}\bar{P})^{1\ldots n}}\left(\rho_{LP(\bar{L}\bar{P})^{(1\ldots n)}}\right) \\[1ex]
        \vspace{1mm}
        \text{(iii)'} \quad \operatorname{tr}_{L}\left[\rho_{LP}\right] = \frac{\mathds{1}_P}{d_P} \\[1ex]
        \vspace{1mm}
        \text{(iv)'} \quad \operatorname{tr}_{\bar{P}}\left[\rho_{\bar{L}\bar{P}}\right] =  \frac{\mathds{1}_{\bar{L}}}{d_{\bar{L}}}.
    \end{array}
\end{align}
In this context, \cite{berta2021semidefinite} showed that $\Sigma_{\operatorname{prod}}^n(LP:\bar{L}\bar{P})$ converges to $\Sigma_{\mathrm{prod}}(LP:\bar{L}\bar{P})$, and therefore to the channel fidelity \eqref{eq:optimization_C}. The key ingredient of their analysis is a quantum de-Finetti theorem that permits linear constraints, in contrast to the DPS hierarchy, whose underlying de-Finetti approximation does not, a priori, allow such constraints. Using the same machinery, we state an analogous convergence result for $\Sigma_{\mathrm{SEP}}(LP:\bar{L}\bar{P})$ with the associated optimization problem
\begin{equation}\label{eq:DPS_type}
\begin{aligned}
    \overline{\operatorname{SDP}}^\star \ \coloneqq \  \sup  \ &\operatorname{tr}[\rho_{LP \bar{L}\bar{P}}C_{\operatorname{id}\otimes \mathcal{N}}] \\
   \operatorname{s.th.}&  \ \rho_{LP\bar{L}\bar{P}} \in \Sigma_{\operatorname{SEP}}(LP:\bar{L}\bar{P}).
\end{aligned}
\end{equation}

\begin{proposition}[Hierarchy of SDP relaxations
\cite{berta2021semidefinite}]\label{thm:hierarchy_outer_approx} \label{prop:dps_type_hierarchy}
With the identification $\rho_{LP\bar{L}\bar{P}}=\rho_{LP(\bar{L}\bar{P})^{(1)}}$ we have that
\begin{enumerate}
\item[(a)] the SDP relaxation of \eqref{eq:resulting_opt} given by \eqref{eq:nlevel_relax} yields a sequence of upper bounds on \eqref{eq:definition_channel_fidelity} converging from above with $n\rightarrow\infty$ as
\begin{align}
\begin{array}{ccc}
     \displaystyle 0\leq \operatorname{SDP}^*_n - F(\mathcal{N},d_L) \leq \frac{\operatorname{poly}(d_L,d_P)}{\sqrt{n}} &\displaystyle 
\end{array}
\end{align}
\item[(b)] the SDP relaxation of $\Sigma_{\operatorname{SEP}}(LP:\bar{L}\bar{P})$ is given by \eqref{eq:DPS_type} yields a sequence of upper bounds and converging from above with $n\rightarrow\infty$ as
\begin{align}
\begin{array}{ccc}
     \displaystyle 0\leq \overline{\operatorname{SDP}}^*_n -\overline{\operatorname{SDP}}^\star \leq \frac{\operatorname{poly}(d_L,d_P)}{\sqrt{n}}. &\displaystyle 
\end{array}
\end{align}
\end{enumerate}
The abbreviation $\operatorname{poly}(d_L,d_P)$ stands a dependence scaling at most polynomial in $d_L$ and $d_P$.
\end{proposition}

The counterexamples from \cite[Exam 4.6]{berta2021semidefinite} show that the hierarchies in \autoref{thm:hierarchy_outer_approx} (a) and (b) are in general strictly different, and for the cSEP problem of AQEC we do need the setting as in \autoref{thm:hierarchy_outer_approx} (a).\footnote{We note that so-called infinite de Finetti theorems with linear constraints can straightforwardly be deduced from standard infinite de Finetti theorem, as first observed by Fuchs \& Schack \cite{Fuchs_2004,Fuchs2004} (as well as others \cite{yu2021complete,PRXQuantum.3.010340,ohst2024characterisingmemoryquantumchannel}). However, this does not imply explicit convergence speed bounds of the corresponding SDP hierarchies. Further, stronger variations for infinite de Finetti theorems with linear constraints were shown by Costa {\it et al.}~\cite{costa2024finettitheoremquantumcausal}, but it remains an open problem if finite refinements thereof are possible.} However, from a computational perspective, implementing the hierarchy in \autoref{thm:hierarchy_outer_approx} (a) poses significant challenges. The number of equality constraints grows exponentially with the number of copies, making direct enforcement infeasible. Specifically, \autoref{thm:hierarchy_outer_approx} (a) involves $d_P^2 (d_L^2 d_P^2)^n + d_L^2 (d_L^2 d_P^2)^{n-1}$ equality constraints, whereas \autoref{prop:dps_type_hierarchy} (b) requires only $d_L^2 + d_P^2$. This contrast highlights the impracticality of the former for large $n$ without symmetry reduction (see \cite{chee2024efficient} and \autoref{sec:symmetry_reduction}).

A key distinction between the hierarchies\,---\,as illustrated in \autoref{fig:pic_for_proof_difference_hierarchies}\,---\,lies in the sequence of operations: applying the partial trace and imposing linear constraints. We end this section by showcasing that even in the restricted case of classical-quantum states, a significant difference between the hierarchies can be proven, showing that the two operations in \autoref{fig:pic_for_proof_difference_hierarchies} do not commute.

\begin{figure}[ht]
    \centering
    \includegraphics[width=0.9\linewidth]{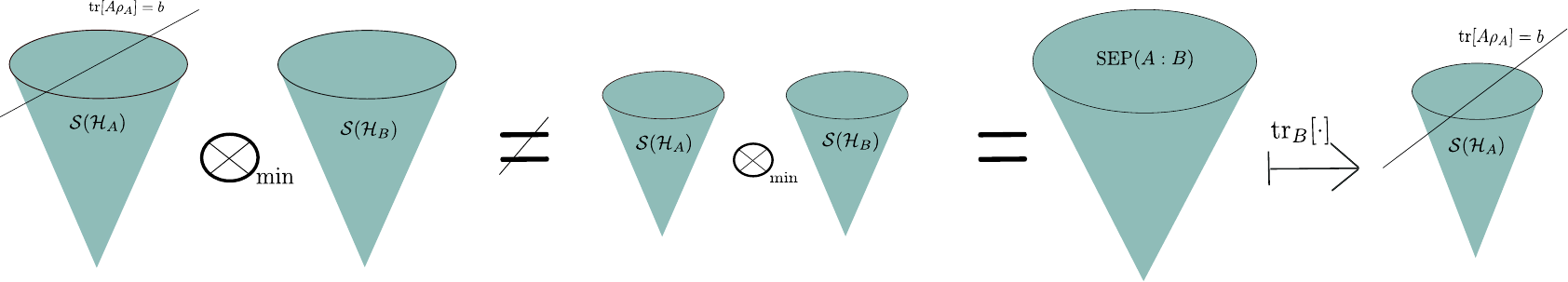}
    \caption{The figure highlights the fundamental difference between the two hierarchies \autoref{thm:hierarchy_outer_approx} (a) and (b). In the first hierarchy, linear constraints are imposed on the cones before applying the minimal tensor product for convexification, that is, the minimal tensor product plays the role of allowing for all convex combinations of products). In the second hierarchy, the minimal tensor product of the cones is constructed first, and constraints are applied afterward. As shown in \autoref{thm:inequality_between_hierarchies}, these operations do not generally commute, underscoring a significant structural divergence.}
    \label{fig:pic_for_proof_difference_hierarchies}
\end{figure}

\begin{proposition}\label{thm:inequality_between_hierarchies}
    Consider the two sets 
    \begin{align}
        \Sigma_1(A:B) \coloneqq  \operatorname{conv} \left\{\rho_{A} \otimes \sigma_{B} \ \middle| \ \rho_{A},\sigma_{B} \in \mathcal{S}(\mathcal{H}_{AB}), \ \operatorname{tr}[\Gamma_A \rho_{A}] = b, \ [\rho_A,\Gamma_A] = 0 \right\}.
    \end{align}
    and 
    \begin{align}
        \Sigma_2(A:B) \coloneqq \{\rho_{AB} \in \operatorname{SEP}(A:B) \ \vert \ \operatorname{tr}[\Gamma_A \rho_A] = b, \ [\rho_A, \Gamma_A]  =0  \},
    \end{align}
    with $\Gamma_A$ a hermitian operator on $\mathcal{H}_A$ not equal to the zero operator and $\lambda_{\operatorname{min}}(\Gamma) < b < \lambda_{\operatorname{max}}(\Gamma)$. Then, we have that $\Sigma_1(A:B)$ and $\Sigma_2(A:B)$ are not equal.
\end{proposition}

\begin{proof}
    Without loss of generality, we add $r\mathds{1}$ on $\Gamma$ for an adequate $r>0$, such that $\Gamma$ becomes positive, $b \mapsto b+r$. From now on we call this operator $\Gamma$ and the corresponding right hand side of the constraint $b$. Due to the fact that states are positive and have trace equal to $1$, $b$ is now positive and 
    \begin{align}
        \lambda_{\operatorname{min}}(\Gamma) \leq b \leq \lambda_{\operatorname{max}}(\Gamma),
    \end{align}
    with even strict inequalities due to the second assumption. Furthermore, even in case of degeneracy of the eigenspaces of $\Gamma$, we consider just one-dimensional projections onto the eigenspaces of $\lambda_{\operatorname{min}}(\Gamma)$ and $\lambda_{\operatorname{max}}(\Gamma)$, from now on called $P_{\operatorname{min}}$ and $P_{\operatorname{max}}$. Out of these, together with two orthogonal rank-$1$ projectors $Q_1 \perp Q_2$ on the $B$-system, we build the state
    \begin{align}\label{eq:first_rep_rho}
       \rho_{AB} \coloneqq  p P_{\operatorname{min}} \otimes Q_1 + (1-p)P_{\operatorname{max}} \otimes Q_2.
    \end{align}
    Since $\rho_{AB}$ belongs to $\Sigma_2(A:B)$, there exists a unique $p$ such that
    \begin{align}\label{eq:conclusion_proof}
        b = p\lambda_{\operatorname{min}} + (1-p)\lambda_{\operatorname{max}} \ \Rightarrow \ p = \frac{\lambda_{\operatorname{max}} - b}{\lambda_{\operatorname{max}} - \lambda_{\operatorname{min}}}.
    \end{align}
    The second assumption then yields that $p \neq 0,1$ such that $\rho_{AB}$ is a separable state but not product. Assume now for contradiction that there exists $q_j\geq 0$ with $\sum_j q_j = 1$ and states $\rho_A^j$ and $\rho_B^j$ from $\mathcal{S}(\mathcal{H}_A)$ and $\mathcal{S}(\mathcal{H}_B)$, respectively, such that
    \begin{align}\label{eq:sec_rep}
        \rho_{AB} = \sum_{j} q_j \rho_A^j \otimes \rho_B^j
    \end{align}
    for states $\rho_A^j$ such that 
    \begin{align}
        \operatorname{tr}[\rho_A^j \Gamma] = b \quad \text{and} \quad [\rho_A^j,\Gamma] = 0 \quad \text{for all}  \ j.
    \end{align}
    The fact that the reduced states of several representations of a state must be equal and that the support of the reduced state from \eqref{eq:first_rep_rho} on the $A$-system is just on the vector space spanned by $\langle P_{\operatorname{min}},P_{\operatorname{max}}\rangle$ leads us to the conclusion that every $\rho_{A}^j$ needs to have support on this subspace. Thus, we conclude, with the commutativity with $\Gamma$, that every element $\rho_A^j$ can be written diagonal in the basis of $\Gamma$
    \begin{align}
        \rho_{A}^j = r_j P_{\operatorname{min}} + (1-r_j) P_{\operatorname{max}}, \quad \text{for some} \ r_j \in [0,1]
    \end{align}
    and has to satisfy the constraint $\operatorname{tr}[\rho_A^j \Gamma] = b$. Together we are back to \eqref{eq:conclusion_proof} which yields that for every $j$ in the decomposition \eqref{eq:sec_rep} we have 
    \begin{align}
        \rho_{A}^j = p P_{\operatorname{min}} + (1-p) P_{\operatorname{max}} \quad \text{for all}  \ j.
    \end{align} 
    Thus, the state $\rho_{AB}$ from \eqref{eq:sec_rep} can be written as
    \begin{align}
        \rho_{AB} = \sum_j q_j (p P_{\operatorname{min}} + (1-p) P_{\operatorname{max}}) \otimes \rho_B^j= (p P_{\operatorname{min}} + (1-p) P_{\operatorname{max}}) \otimes \sum_j q_j \rho_B^j.
    \end{align}
    However, this is a product state and thus a contradiction to the fact that $\rho_{AB}$ is not product by construction. We conclude that states of the form considered have no separable decomposition into a form necessary for the set $\Sigma_1(A:B)$.
\end{proof}


\section{Inner bounds}\label{sec:inner_bounds}

\subsection{Rounding to inner points}

In this section, we aim to bridge the gap between the see-saw methods for \eqref{eq:resulting_opt} from \cite{Reimpell_2005} and the SDP hierarchies discussed in \autoref{subsec:definition_of_hierarchies}. While \autoref{thm:hierarchy_outer_approx} (a) and (b) provide outer bounds for related optimization problems, see-saw methods produce inner bounds, yielding practically implementable error-correcting codes.  From a broader perspective, connecting these approaches is a familiar challenge in non-convex bilinear optimization. A foundational example is the Goemans-Williamson approximation for quadratic unconstrained binary optimization problems, which pioneered a rounding technique. In this method, an outer solution is converted into an inner solution through probabilistic sampling over a sphere in $\mathbb{R}^{n+1}$ for an $n$-dimensional problem instance. We construct a quantum analogue of this probabilistic rounding process, seamlessly integrating it into the framework of quantum measurements. Conceptually, our approach is illustrated in \autoref{fig:inner_hierarchy}.

The process involves projecting an outer solution obtained from \autoref{thm:hierarchy_outer_approx} (a) onto an inner point, ensuring that the constraints are satisfied and the objective value remains controlled. This step mirrors the proof of the outer hierarchy, where the projections shown in \autoref{fig:inner_hierarchy} correspond to quantum measurements.

\begin{figure}[ht!]
    \centering
    \includegraphics[width=0.5\linewidth]{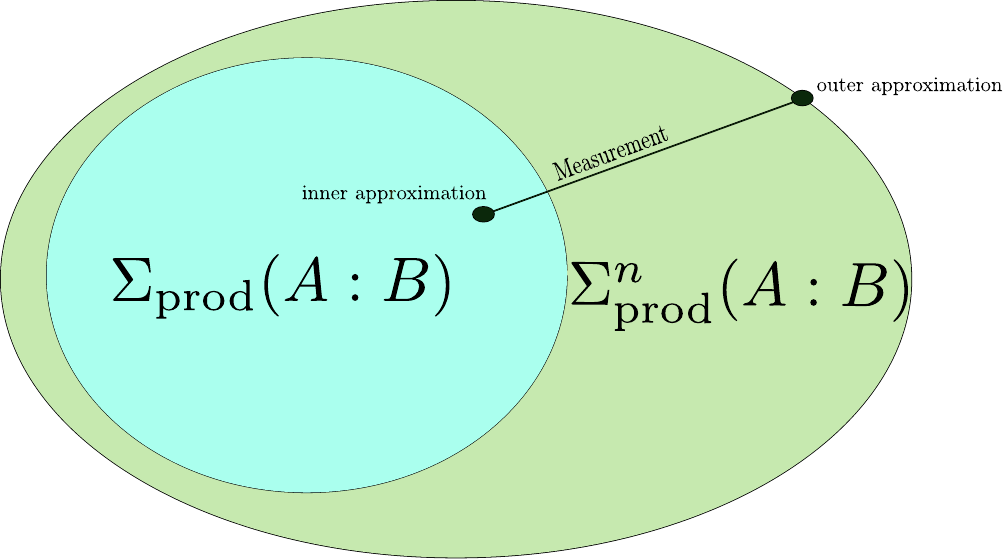}
    \caption{This diagram illustrates the process of obtaining an inner point. Starting with a solution from the outer hierarchy in \autoref{thm:hierarchy_outer_approx}, which lies within the set of $n$-extendable states, we project this solution onto an inner point. The projection ensures that the resulting state satisfies the constraints while maintaining control over the objective value.}
    \label{fig:inner_hierarchy}
\end{figure}

The following theorem establishes the conditions necessary to obtain these inner points, which correspond to viable error-correcting codes.

\begin{theorem}[Inner bounds]\label{thm:inner_bounds}
    From the $n$-th level optimizers
    \begin{equation}\label{eq:optimizer_thm_inner_bounds}
    \begin{aligned}
        \rho^\star_{LP (\bar{L}\bar{P})^{(1\ldots n)}} = \arg \max \ &d_P^2\operatorname{tr}\left[C_{(\operatorname{id}_L \otimes \mathcal{N})} \rho_{LP\bar{L}\bar{P}}\right]&\\
        \operatorname{s.th.} \ & \rho_{LP\bar{L}\bar{P}} \in \Sigma_{\operatorname{prod}}^n(LP:\bar{L}\bar{P}).
    \end{aligned}
    \end{equation}
    one can construct a sequence of inner points $\sigma^{[n]}_{LP\bar{L}\bar{P}} \in \Sigma_{\operatorname{prod}}(LP:\bar{L}\bar{P})$ with the same performance guarantee \autoref{thm:hierarchy_outer_approx} (a) as of the outer hierarchy.   
\end{theorem}

\begin{proof}
Let $\mathcal{M}_{\bar{L}\bar{P}\to Z}$ be an informationally complete measurement on $\bar{L}\bar{P}$ into a classical register $Z$, i.e.
    \begin{align}
        \Vert\zeta_{LP \bar{L}\bar{P}}\Vert_1 \leq c(d_{\bar{L}\bar{P}}) \Vert \mathcal{M}_{\bar{L}\bar{P}\to Z}(\zeta_{LP\bar{L}\bar{P}})\Vert_1
    \end{align}
for all hermitian operators $\zeta_{LP \bar{L}\bar{P}}$ on $\mathcal{H}_{LP} \otimes \mathcal{H}_{\bar{L}\bar{P}}$. Following the proof strategy of \cite[Thm. 2.3]{berta2021semidefinite}, there exists an $m$-fold product measurement into classical registers $Z_2,\ldots,Z_{m+1}$
\begin{align}
    \mathcal{M}_{(\bar{L}\bar{P})^{(2)} \to Z_2} \otimes \cdots \otimes \mathcal{M}_{(\bar{L}\bar{P})^{(m+1)} \to Z_{m+1}},
\end{align}
which yields local post-measurement states conditioned on outcomes $z_2^{m+1}$
\begin{align}
    \rho^\star_{LP \vert z_2^{m+1}} = \frac{\tr_{(\bar{L}\bar{P})^{(1)} Z_2^{m+1}}[\mathcal{M}_{(\bar{L}\bar{P})^{(2)} \to Z_2}^{z_2} \otimes \cdots \otimes \mathcal{M}_{(\bar{L}\bar{P})^{(m+1)} \to Z_{m+1}}^{z_{m+1}}(\rho^\star_{LP (\bar{L}\bar{P})^{(1\ldots m+1)}})]}{\tr[\mathcal{M}_{(\bar{L}\bar{P})^{(2)} \to Z_2}^{z_2} \otimes \cdots \otimes \mathcal{M}_{(\bar{L}\bar{P})^{(m+1)} \to Z_{m+1}}^{z_{m+1}}(\rho^\star_{LP (\bar{L}\bar{P})^{(1\ldots m+1)}})]}
\end{align}
and similar for $\rho^\star_{\bar{L}\bar{P} \vert z_2^{m+1}}$. Together we let 
\begin{align}
   \sigma^{[n]}_{LP\bar{L}\bar{P}} \coloneqq \sum_{z_2^{m+1}}p(z_2^{m+1}) \rho^\star_{LP \vert z_2^{m+1}} \otimes \rho^\star_{\bar{L}\bar{P}\vert z_2^{m+1}},
\end{align}
which satisfies the constraints from \autoref{thm:hierarchy_outer_approx} (a) and fulfills due to \cite[Thm. 2.3]{berta2021semidefinite}
\begin{align}\label{eq:proof_inner_bound_estimate}
        \Vert \rho^\star_{LP \bar{L}\bar{P}} - \sigma^{[n]}_{LP \bar{L}\bar{P} }\Vert_1 \leq c(d_{\bar{L}\bar{P}}) \sqrt{\frac{2 \ln (2) \ln (d_{\bar{L}\bar{P}})}{n}}.
\end{align}
 Now we consider the corresponding objective function value $d_{\bar{L}\bar{P}}^2\operatorname{tr}[\sigma^{[n]}_{LP\bar{L}\bar{P}}C_{\operatorname{id}_L \otimes \mathcal{N}}]$ of the channel fidelity from \eqref{eq:resulting_opt} in comparison to an $n$-level relaxation and use \eqref{eq:proof_inner_bound_estimate} in order to obtain
 \begin{equation} \label{eq:estimate_bound_inner_proof_objective}
\begin{aligned}
    &\mathcal{F}(\mathcal{N},d_L) - d_{\bar{L}\bar{P}}^2\operatorname{tr}[C_{\operatorname{id}_L \otimes \mathcal{N}} \sum_{z_2^{m+1}}p(z_2^{m+1}) \rho^\star_{LP \vert z_2^{m+1}} \otimes \rho^\star_{\bar{L}\bar{P}\vert z_2^{m+1}}]\\
    &\leq \operatorname{SDP}_n^\star - d_{\bar{L}\bar{P}}^2\operatorname{tr}[C_{\operatorname{id}_L \otimes \mathcal{N}} \sum_{z_2^{m+1}}p(z_2^{m+1}) \rho^\star_{LP \vert z_2^{m+1}} \otimes \rho^\star_{\bar{L}\bar{P}\vert z_2^{m+1}}] \\
    &= d_{\bar{L}\bar{P}}^2\operatorname{tr}[C_{\operatorname{id}_L \otimes \mathcal{N}} \rho^\star_{LP\bar{L}\bar{P}}] - d_{\bar{L}\bar{P}}^2\operatorname{tr}[C_{\operatorname{id}_L \otimes \mathcal{N}} \sum_{z_2^{m+1}}p(z_2^{m+1}) \rho^\star_{LP \vert z_2^{m+1}} \otimes \rho^\star_{\bar{L}\bar{P}\vert z_2^{m+1}}]\\
    &\leq d_{\bar{L}\bar{P}}^2 \Vert C_{\operatorname{id}_L \otimes \mathcal{N}}\Vert_\infty \Vert \rho^\star_{LP \bar{L}\bar{P}} - \sigma^{[n]}_{LP \bar{L}\bar{P} }\Vert_1 \\
    &\leq \Vert C_{\operatorname{id}_L \otimes \mathcal{N}}\Vert_\infty d_{\bar{L}\bar{P}}^2 c(d_{\bar{L}\bar{P}}) \sqrt{\frac{2 \ln (2) \ln (d_{LP})}{n}}.
\end{aligned}
\end{equation}
In the second step we used Hölder's inequality for the Schatten-$1$ and Schatten-$\infty$ norm. One can even obtain a better product state, if one maximizes over all measurement outcomes 
\begin{align*}
    \max_{z_2^{m+1}}\operatorname{tr[C_{\operatorname{id}_L \otimes \mathcal{N}} \rho_{LP \vert z_2^{m+1}} \otimes \rho_{\bar{L}\bar{P}\vert z_2^{m+1}}]},
\end{align*}
which is an inner point due to construction, and all estimates from \eqref{eq:estimate_bound_inner_proof_objective} are still valid. Because the number $2\leq m \leq n$ is unknown, we iterate through all possible values of $m$ in order to get the desired bound. Thus, we constructed a sequence of values converging from below towards $\mathcal{F}(\mathcal{N},d_L)$.
\end{proof}

The core insight of \autoref{thm:inner_bounds} is that for the state $\rho_{LP (\bar{L}\bar{P})^{(1\ldots n)}}$ the constraint structure is designed such that post-measurement states\,---\,obtained by measuring the extended systems\,---\,satisfy the Choi constraints. Importantly, Choi constraints are partial trace constraints and break the correlations in the constraints such that an adaptation of \cite[Theorem 9]{jee2020quasi} to the bipartite setting follows directly by \autoref{cor:better_bound_mutual_info}, which follows the lines of \cite[Lemma 6]{jee2020quasi}. Thus, tight bounds on the mutual information can improve the constants hidden in $\operatorname{poly}(d_L,d_P)$.


\subsection{Practical considerations}

To construct the inner points, we need access to a suitable measurement on the extended systems. From the proof strategy, a small distortion factor $c(d_{\bar{L}\bar{P}})$ of the informationally complete measurement $\mathcal{M}$
\begin{align}
    \Vert\zeta_{LP \bar{L}\bar{P}}\Vert_1 \leq c(d_{\bar{L}\bar{P}}) \Vert \mathcal{M}(\zeta_{LP\bar{L}\bar{P}})\Vert_1
\end{align}
is needed in order to achieve a controllable performance guarantee as of the outer hierarchy. In that respect, an optimal construction was reported in \cite{jee2020quasi}, scaling as $c(d_{\bar{L}\bar{P}})=O(d_{\bar{L}\bar{P}})$. Practically, however, controlling the number of measurement outcomes is also critical, as it directly affects the complexity of the search algorithm over all outcomes. However, for informationally complete measurements on the $\bar{L}\bar{P}$-system, the number of measurement outcomes for all copies scales exponentially in the number of copies. Hence, from a practical point of view it may be of significant interest to develop problem adapted measurements with significantly fewer outcomes\,---\,but lacking a priori analytical guarantees (as they might not be informationally complete anymore). Here we note that any measurement will give valid error correction codes, just without the rigorous optimality bounds on their performance.

To conclude, the procedure for generating inner bounds is as follows:
\begin{enumerate}
    \item Solve the outer hierarchy from \autoref{thm:hierarchy_outer_approx} (a).
    \item Choose an appropriate set of measurements and construct a separable post-measurement state as per \autoref{thm:inner_bounds}.
    \item Optimize over all conditional states $ \rho_{LP \vert z_2^{m+1}} \otimes \rho_{\bar{L}\bar{P} \vert z_2^{m+1}} $ with respect to the objective value.
    \item Optionally, perform a few iterations of the see-saw algorithm and compare the gap between the outer and inner values.
\end{enumerate}


\section{Dimension reduction via symmetries}
\label{sec:symmetry_reduction}

\subsection{Introduction}

In this section, we examine the task of symmetry reduction for semi-definite programs (SDPs), focusing specifically on its application to approximate quantum error correction (AQEC). Our motivation stems from the inherent permutation symmetry in the definition of $n$-extendable states and the need to use \autoref{thm:hierarchy_outer_approx} (a), which exhibits a more intricate permutation symmetry than the standard DPS hierarchy and, by extension, \autoref{prop:dps_type_hierarchy} (b). Indeed, a rough estimate of the dimension of the final optimization problem shows that if one encodes a single qubit into $m$ qubits and requires $n$ extensions, we have
\begin{align}\label{eq:dimension_estimate}
    \operatorname{dim} \mathcal{H}_{LP(\bar{L}\bar{P})_1^n} 
    = 2^{1 + m +  n\cdot (1+ m)}
    \quad \text{as the overall dimension.}
\end{align}
We recall from \cite{Leung_2015,berta2021semidefinite} that extensions are needed to circumvent the relatively simple non-signaling assistant channel fidelity bound (which is typically far off as it still includes powerful entanglement-assistance). Further, for $m=2$ it is well-known that a trivial encoder--decoder pair often already achieves the optimal channel fidelity, so the first genuinely non-trivial case arises when encoding one qubit into three qubits. Consequently, for $n=2$ and $m=3$, the overall dimension in \eqref{eq:dimension_estimate} is already $4096$. Moreover, from the discussion in \autoref{subsec:definition_of_hierarchies}, we note that the number of equality constraints grows polynomial with $\operatorname{dim}(\mathcal{H}_L)$ and $\operatorname{dim}(\mathcal{H}_P)$ but exponentially in $n$ and $m$. The seminal work by Burer-Monteiro \cite[Thm.~2.2]{BurerMonteiro2003} then indicates that one should expect a dense, high-rank solution of the optimization problem in \autoref{thm:hierarchy_outer_approx} (a). Hence, in a generic scenario, we may well face an unstructured (if we disregard symmetrization) and dense SDP whose block dimension on the order of $1000$--$2000$ is already quite challenging to solve, often requiring HPC resources.\footnote{Indeed, each interior-point iteration typically costs on the order of $\mathcal{O}\bigl((\operatorname{dim}\mathcal{H})^3\bigr)$ floating-point operations, and one typically needs tens of such iterations. Of course, the worst-case bound of $\mathcal{O}\bigl(\sqrt{N}\log\frac{1}{\varepsilon}\bigr)$ iterations also applies, where $N$ scales with the problem size and $\varepsilon$ is the target precision, but in practice this detailed estimate is often unnecessary.} As such, the dimension estimate in \eqref{eq:dimension_estimate} and the resulting dense SDP underscore the need to reduce the dimension in order to bridge the gap between the theoretically sound results of \cite{berta2021semidefinite} and practical applications to non-trivial examples.

In a recent work of Chee {et al.}~\cite{chee2024efficient}, the extendability symmetry for AQEC was worked out in the standard framework of symmetry reduction \cite{Klerk2007ReductionOS}, yielding practical results for the task of encoding a single qubit into a single qubit (see \cite[Table~1]{chee2024efficient}). However, scaling up the number of physical qubits may demand more advanced methods and we illustrate this by considering the coding task of mapping one qubit to $m$ qubits. As previously mentioned, to avoid the non-signaling assistant bound, we aim to have at least one extension; hence, we let $n = 2$. Then the total dimension of the problem with encoding one qubit into $m$ qubit is
\begin{align}
  \dim\bigl(\mathcal{H}_{LP(\bar{L}\bar{P})_1^2}\bigr)
  \;=\;
  2^{3 \cdot (1 + m)}.
\end{align}
Meanwhile, the symmetrization from the extendability symmetry acts only on a subspace of dimension $2^{2(1+m)}$, governed by an $S_2$ symmetry. The symmetric group $S_2$ has two irreducible representations of dimension $1$ (for details see e.g. \cite{Sagan2001}), with multiplicities given by the sizes of the symmetric and antisymmetric subspaces
\begin{equation}
  \begin{aligned}
    \dim\!\bigl(\mathrm{Sym}^2(\mathcal{H}_{LP})\bigr) 
    &= \frac{2^{m+1}\bigl(2^{m+1} + 1\bigr)}{2}\\[6pt]
    \dim\!\bigl(\wedge^2 \mathcal{H}_{LP}\bigr) 
    &= \frac{2^{m+1}\bigl(2^{m+1} - 1\bigr)}{2}.
  \end{aligned}
\end{equation}
Concretely, for $m=3$, we find
\begin{align}
  \dim\!\bigl(\mathrm{Sym}^2(\mathcal{H}_{LP})\bigr) \;=\; 136
  \quad\text{and}\quad
  \dim\!\bigl(\wedge^2(\mathcal{H}_{LP})\bigr) \;=\; 120.
\end{align}
Multiplying these by the dimension of the first system, $2^{3+1} = 16$, which is unaffected by the standard symmetry reduction in \cite{chee2024efficient} regarding the extendibility symmetry, yields two blocks of sizes
\begin{equation}
  \begin{aligned}
    16\cdot \dim\!\bigl(\mathrm{Sym}^2(\mathcal{H}_{LP})\bigr) 
    &= 2176, \\
    16\cdot \dim\!\bigl(\wedge^2(\mathcal{H}_{LP})\bigr) 
    &= 1920.
  \end{aligned}
\end{equation}
As already discussed, blocks of this size can already be quite challenging to handle on standard hardware.

In the following, we carefully examine whether additional symmetries\,---\,as outlined in \autoref{subsec:types_of_symmetries}\,---\,can be leveraged to further reduce the dimension. A key challenge arises because imposing certain symmetries in the objective function of \autoref{thm:hierarchy_outer_approx} (a), specifically, in the Choi matrix $C_{\mathrm{id} \otimes \mathcal{N}}$ for the noisy channel
\begin{align}
  \mathcal{N} : \mathcal{S}(\mathcal{H}_P) \to \mathcal{S}(\mathcal{H}_P)
\end{align}
may not be straightforwardly compatible with the extendibility symmetry.
 
\begin{example}\label{example:all_symmetries}
Suppose we have three physical qubits and the task is to encode one logical qubit within them. In this case, the setup \eqref{eq:setting_error_correction} is
\begin{align}
    \mathcal{S}(\mathcal{H}_L) \stackrel{\mathcal{E}}{\longrightarrow} \mathcal{S}\left(\bigotimes_{i=1}^3\mathcal{H}_{P_i}\right) \stackrel{\mathcal{N}}{\longrightarrow} \mathcal{S}\left(\bigotimes_{i=1}^3\mathcal{H}_{P_i}\right) \stackrel{\mathcal{D}}{\longrightarrow} \mathcal{S}(\mathcal{H}_L),
\end{align}
where $ \mathcal{H}_{P_i}$, $1\leq i \leq 3$, is isomorphic to $ \mathbb{C}^2 $ as a vector space. Assume further that the noisy channel is a depolarizing channel, represented by
\begin{align}
\mathcal{N} \equiv \mathcal{M}_{\operatorname{dep}}^{\otimes 3}: \rho \mapsto \bigotimes_{i=1}^3 \left( (1-p) \mathds{1}_i \rho \mathds{1}_i + \frac{p}{3}(X_i \rho X_i + Y_i \rho Y_i + Z_i \rho Z_i) \right),
\end{align}
where $ X_i $, $ Y_i $, and $ Z_i $ act only on the $ i $-th subsystem  $ \mathcal{H}_{P_i}$. This noise model introduces the two additional symmetries
\begin{itemize}
    \item \autoref{subsec:types_of_symmetries} (c) iid-symmetry over the three qubits, which corresponds to the symmetric group $ S_3 $
    \item \autoref{subsec:types_of_symmetries} (d) Choi symmetry via the fact that the depolarizing channel commutes with all single-qubit unitary channels (cf.~\autoref{def:covariant_channels}).
\end{itemize}
The covariance of the depolarizing channel with the unitary group results in a particularly simple form for its Choi matrix 
\begin{align}
C_{\mathcal{M}_{\operatorname{dep}}} = \frac{1-p}{2} \mathds{1} \otimes \mathds{1} + \frac{p}{6} \left( X \otimes X + Y \otimes Y + Z \otimes Z \right).
\end{align}
Thus, the Choi matrix of the depolarizing channel is invariant under the action of the unitary group on two qubits, or equivalently, it is an isotropic state. Further, if we recap the channel fidelity in \eqref{eq:resulting_opt}, we get a tensor product with a maximally entangled state as the Choi state for the identity channel, which is another isotropic state symmetry, \autoref{subsec:types_of_symmetries} (b), and thus satisfies invariance under the local unitary group. We then aim to combine the permutation symmetry arising from repeated applications of the local noisy channel, called \emph{iid-symmetry}, with the local covariance symmetry of the depolarizing channel, called \emph{Choi symmetry}, and the \emph{isotropic state symmetry}\,---\,as well as the permutation symmetry inherent to the separability problem, called \emph{extendibility symmetry}.
\end{example}


\subsection{Standard framework}
\label{subsec:recap_standard_symmetry}

Next, we introduce our notation within the standard symmetry reduction framework. We consider SDPs of the form
\begin{equation}\label{eq:general_SDP}
    \begin{aligned}
        c^\star \coloneqq   \sup \ &\operatorname{tr}[\rho H] \\
         \operatorname{s.th.} \ &\operatorname{tr}[\rho A_j] \leq b_j \\
         &\rho \geq 0.
    \end{aligned}
\end{equation}
We say that an SDP as \eqref{eq:general_SDP} with value $c^\star$ is invariant with respect to a unitary symmetry group $\mathcal{W}$ as a subgroup of the unitary group $\mathcal{U}(\mathcal{H})$ if for all $U \in \mathcal{W}$
\begin{enumerate}
    \item $U H U^\dagger = H$ 
    \item $U A_j U^\dagger = A_j$ for all $j$.
\end{enumerate}
Let $\mu$ be the Haar measure corresponding to $\mathcal{W}$, then a simple argument about the cyclicity of the trace shows that for every solution $\rho^\star \in \mathcal{S}(\mathcal{H})$, which is a solution to \eqref{eq:general_SDP}, can be mapped under the channel 
\begin{align}
   T: \rho \mapsto \int_\mathcal{W} d\mu(U) U \rho U^\dagger
\end{align}
and $T(\rho^\star)$ has the same value as $\rho^\star$, but is now invariant under $\mathcal{W}$. Hence, we can write
\begin{equation}\label{eq:commutant}
    \begin{aligned}
        c^\star \coloneqq   \sup \ &\operatorname{tr}[\rho H] \\
     \operatorname{s.th.} \ &\operatorname{tr}[\rho A_j] \leq b_j \quad 1\leq j \leq n\\
         &\rho \geq 0, \quad T(\rho) = \rho.
    \end{aligned}
\end{equation}
We then consider $\{U\}_{U\in \mathcal{W}}$ as a generating set for a group algebra $\mathcal{M}$. Given $\mathcal{M}$, our goal is to identify an appropriate unitary transformation such that we can represent $\mathcal{M}$ in terms of a direct sum structure. Specifically, we aim to construct
\begin{align}
    \mathcal{H} \to \bigoplus_{i=1}^k \mathbb{C}^{d_i \cdot m_i},
\end{align}
where the mapping transforms $\mathcal{H}$ into a block-diagonal form composed of subspaces of dimension $d_i \cdot m_i$. This decomposition provides a powerful way to parametrize the commutant $\mathcal{M}^\prime$. Since $\mathcal{M}^\prime$ consists of operators that commute with $\mathcal{M}$, it takes on a complementary structure in terms of a direct sum of matrix algebras. Specifically, we can assume the existence of an isometry $V$ such that every element $m \in \mathcal{M}^\prime$ can be mapped to an element in a direct sum of matrix algebras
\begin{align}\label{eq:decompose_mprime}
    \mathcal{M}^\prime \to \bigoplus_{i=1}^k \mathbb{C}^{m_i \times m_i}, \quad m \mapsto V^\dagger m V,
\end{align}
where $V^\dagger m V$ parametrizes elements of the commutant $\mathcal{M}^\prime$ within the given direct sum structure. By applying this map to the operators $H, A_1, \ldots, A_n$, we can reformulate the original optimization problem as the symmetry-reduced problem
\begin{equation}\label{eq:symmetry_reduced_problem}
\begin{aligned}
    c^\star = \sup \ & \sum_{i=1}^k \operatorname{tr}[X_i \tilde{H}_i] \\
    \operatorname{s.th.} \ &\sum_{i=1}^k \operatorname{tr}[X_i \tilde{A}_{ij}] \leq b_j \quad 1 \leq j \leq n, \\
    &X_i \geq 0, \quad 1 \leq i \leq k.
\end{aligned}
\end{equation}
In this reduced problem, $\tilde{H}_i$ and $\tilde{A}_{ij}$ denote the images of the original operators $H$ and $A_j$ under the map $V^\dagger \cdot V$, resulting in a decomposition into $k$ distinct components. Given that $V^\dagger V = \operatorname{id}_{\mathcal{H}}$, the map $V^\dagger \cdot V$ is unital, ensuring that its dual preserves the trace. This implies that the values of the original problem in \eqref{eq:general_SDP} and the reduced problem in \eqref{eq:symmetry_reduced_problem} remain equal.


\subsection{Extendibility symmetry}
\label{subsec:symmetry_reduction_separability}

Given our standard setting 
\begin{align}
    \mathcal{S}(\mathcal{H}_L) \stackrel{\mathcal{E}}{\longrightarrow} \mathcal{S}(\mathcal{H}_P) \stackrel{\mathcal{N}}{\longrightarrow} \mathcal{S}(\mathcal{H}_P) \stackrel{\mathcal{D}}{\longrightarrow} \mathcal{S}(\mathcal{H}_L)
\end{align}
the objective is given by $C_{\operatorname{id}_L \otimes \mathcal{N}} \in \mathcal{B}(\mathcal{H}_L \otimes \mathcal{H}_P \otimes \mathcal{H}_{\bar{L}} \otimes \mathcal{H}_{\bar{P}})$. Moreover, as introduced in \autoref{subsec:types_of_symmetries} (c), we have the iid-symmetry, whereby we consider a channel $\mathcal{M}:\mathcal{S}(\mathcal{H}_{P_i}) \to \mathcal{S}(\mathcal{H}_{P_i})$ and identify $\mathcal{N} \equiv \mathcal{M}^{\otimes m}$. Then, the Choi matrix can be rewritten as $C_{\operatorname{id}_L \otimes \mathcal{M}^{\otimes m}}$. After permuting the systems, this results in the notation of \autoref{fig:symmetry}, whereby extensions just get a label ${}^{(i)}$ for the $i$-th extended system.\\

A dimension reduction of the size of the SDP above based on the extendability symmetry was already done in \cite{chee2024efficient} with tools coming from \cite{litjens2017semidefinite}. In this section, we add a slightly different point of view on the topic, i.e. a primal symmetry reduction, and show how our approach enables us to incorporate entanglement witnesses as constraints as well. Additionally we add a discussion about what happens when we examine the space $\mathcal{M}^\prime$ associated with the group algebra of the symmetric group $S_n$, denoted by $\mathcal{M}$. In particular, it is crucial to analyze how the mappings $\operatorname{tr}[C_{\operatorname{id}_L \otimes \mathcal{N}} \cdot]$, $\operatorname{tr}_P[\cdot]$, and $\operatorname{tr}_L[\cdot]$ interact with this space.\\

Our objective is to deconstruct $\mathcal{M}^\prime$ into a direct sum, as demonstrated in \eqref{eq:decompose_mprime}. This allows us to reformulate the optimization problem accordingly and directly apply the constraints to this direct sum structure. Evaluating the objective function is straightforward, as we can leverage the isometry $V$ obtained from the established methods of symmetry reduction and apply it to $C_{\operatorname{id}_L \otimes \mathcal{N}}$. For the constraints, we will consider the following sequence of mappings
\begin{align}
  \Psi_1: \mathcal{B}(\mathcal{H}_{LP})\otimes  \bigoplus_{i=1}^k \mathbb{C}^{m_i \times m_i} \to \mathcal{B}(\mathcal{H}_{LP} \otimes \mathcal{H}_{\bar{L}\bar{P}}^{\otimes  n}),
\end{align}
the map
\begin{align}
  \operatorname{tr}_L[\cdot] : \mathcal{B}(\mathcal{H}_{LP} \otimes \mathcal{H}_{\bar{L}\bar{P}}^{\otimes  n}) \to \mathcal{B}(\mathcal{H}_P \otimes \mathcal{H}_{\bar{L}\bar{P}}^{\otimes  n}),
\end{align}
and then back to the commutant
\begin{align}
 \tilde{\Psi}_2:  \mathcal{B}(\mathcal{H}_P \otimes \mathcal{H}_B^{\otimes  n}) \to \mathcal{B}(\mathcal{H}_P)\otimes  \bigoplus_{i=1}^k \mathbb{C}^{m_i \times m_i}.
\end{align}
This concatenation of maps is a completely positive map and thus obtains a Kraus-representation. A bit more involved but similarly we can calculate this concatenation of maps for the other linear constraint
\begin{align}
  \Psi_1: \bigoplus_{i=1}^k \mathbb{C}^{m_i \times m_i} \to \mathcal{B}( \mathcal{H}_{\bar{L}\bar{P}}^{\otimes  n}),
\end{align}
the map
\begin{align}
  \operatorname{tr}_{\bar{P}^{(n)}}[\cdot] : \mathcal{B}( \mathcal{H}_{\bar{L}\bar{P}}^{\otimes  n}) \to \mathcal{B}(\mathcal{H}_{\bar{L}\bar{P}}^{\otimes  n-1}\otimes \mathcal{H}_{\bar{L}}),
\end{align}
and then back to the commutant
\begin{align}
 \tilde{\Psi}_2:  \mathcal{B}(\mathcal{H}_{\bar{L}\bar{P}}^{\otimes  n-1}\otimes \mathcal{H}_P) \to  \bigoplus_{j=1}^l \mathbb{C}^{n_j \times n_j} \otimes \mathcal{B}(\mathcal{H}_{\bar{L}^{(n)}}).
\end{align}
Here, the decomposition on the left hand side is a direct sum of matrices w.r.t. the irreducible representations of $S_{n-1}$. Observe that we can control the reduced block-decomposition w.r.t. $S_{n-1}$ with the help of \emph{branching rules} (see e.g. \cite[sec. 2.8]{Sagan2001}). Calculating the Kraus-representation for these maps is computationally similar expensive as calculating the inner sequence, because for the inner sequence we basically need to calculate $\Psi_1$ on a solution $(X_i)_{i=1}^k$.\\

Interestingly, this method works for arbitrary linear constraints, not necessary completely positive ones. This includes for example the positive partial transpose (PPT) entanglement witness. More generally, a state is separable if it is positive under all positive but not completely positive maps $\Gamma:\mathcal{B}(\mathcal{H}_{\bar{L}\bar{P}}) \to \mathcal{B}(\mathcal{H}_{LP})$ \cite{Horodecki_1996,Horodecki_2009}. Thus, for all these positive maps $\Gamma$ we can construct the mapping 
\begin{align}\label{eq:symmetry_reduction_1}
    \Psi_1 \circ \operatorname{tr}_{(\bar{L}\bar{P})^{(2\ldots n)}}[\cdot] \circ \Gamma 
\end{align}
and demand that
\begin{align}\label{eq:symmetry_reduction_2}
     (\Psi_1 \circ \operatorname{tr}_{B_2^n}[\cdot] \circ \Gamma) (\rho_{LP \bar{L}\bar{P}})\geq 0
\end{align}
as an SDP constraint. Importantly, the corresponding Kraus-representation can be chosen with two Hilbert-Schmidt orthogonal families of Kraus operators $\{A_j\}$ and $\{B_j\}$, but they must not coincide, because $(\Psi_1 \circ \operatorname{tr}_{(\bar{L}\bar{P})^{(2\ldots n)}}[\cdot] \circ \Gamma)$ is not completely positive. Thus we have
\begin{equation}
\begin{aligned}
    (\Psi_1 \circ \operatorname{tr}_{(\bar{L}\bar{P})^{(2\ldots n)}}[\cdot] \circ \Gamma): \mathcal{B}(\mathcal{H}_{LP})\otimes  \bigoplus_{i=1}^k \mathbb{C}^{m_i \times m_i}  &\to \mathcal{B}(\mathcal{H}_{LP} \otimes \mathcal{H}_{\bar{L}\bar{P}}) \\ 
    (X_i)_i &\mapsto \sum_j A_j (X_i)_i B_j^\dagger.
\end{aligned}
\end{equation}
We emphasize that even just the PPT constraint can be of strong numerical advantage for the separability problem (cf.~the PPT discussions in \cite{berta2021semidefinite} and results regarding AQEC in \cite{Leung_2015}).\\

We summarize the necessary steps for effectively implementing the evaluation of the program in practice. To compute the mappings $\Psi_i \circ \operatorname{tr}_{\cdot}[\cdot] \circ \tilde{\Psi}_i$, we first need to determine the isometry  $V$ and the corresponding Kraus operators for the partial trace on the larger space. Assuming that basic matrix operations can be performed efficiently in these large-dimensional spaces, this process can be executed relatively fast, although the actual time will naturally depend on the specific dimensions involved. The calculation of the Choi matrix and the Choi operators requires solving an eigenvalue problem for matrices of low rank. Fortunately, this task can also be accomplished in significantly high dimensions, well beyond the operational limits of typical SDP solvers. Ultimately, we are left with a Kraus representation of a channel acting on the smaller matrices $(X_i)_{i=1}^k$, which can be incorporated as constraints within the SDP solver\footnote{We emphasize that the ideas presented in this section are not concerned with the efficient preparation of SDPs in terms of SDP variables, as done in \cite{chee2024efficient}, but rather focus on the conceptual integration of entanglement witnesses into the symmetry reduction framework.}.


\section{Symmetry reduction for AQEC under symmetric noise}\label{sec:commutative_relaxation}

Continuing the notation of \autoref{example:all_symmetries} in this section, we take $\mathcal{N} = \mathcal{M}^{\otimes m}$, the $m$-fold product of a noisy channel $\mathcal{M}$ on subsystems $P_i$ of $P$ into itself such that $\bigotimes_{i=1}^m \mathcal{H}_{P_i} \cong \mathcal{H}_P$ as vector spaces. Consequently, we have the iid-symmetry in the Choi matrix $C_{\operatorname{id}_L \otimes \mathcal{N}}$ from \autoref{subsec:types_of_symmetries} (c), which has consequences for all copies of $\mathcal{H}_P\otimes \mathcal{H}_L$. 

Similarly to \autoref{subsec:symmetries_of_choi}, we first apply a global transformation such that
\begin{align}
   C_{\mathcal{M}^{\otimes m} \otimes \operatorname{id}_L} \mapsto C_{\operatorname{id}_L} \otimes C_{\mathcal{M}}^{\otimes m},
\end{align} 
which motivates us to construct the problem in the language of the Hilbert space 
\begin{align}
\bigotimes_{i=1}^{m} \mathcal{H}_{P_i} \otimes \bigotimes_{i=1}^{m}  \mathcal{H}_{\bar{P}_i} \otimes \mathcal{H}_{L} \otimes \mathcal{H}_{\bar{L}}.    
\end{align}
The bilinear optimization of \eqref{eq:resulting_opt} can be rewritten as 
\begin{equation}
\begin{aligned}\label{eq:optimization_channel_fidelity_symmetry}
       \mathcal{F}(\mathcal{M}^{\otimes m},d_L) &\coloneqq   d_{P_{1\ldots m}} d_{\bar{P}_{1\ldots m}} \max \operatorname{tr}[C_{\mathcal{M}}^{\otimes m} \otimes C_{\operatorname{id}_L} \rho_{P_{1\ldots m} \bar{P}_{1\ldots m}L\bar{L}}]& \\ 
        \operatorname{s.th.} \ &\rho_{P_{1\ldots m} \bar{P}_{1\ldots m}L\bar{L}} \quad \text{is product over} \quad P_{1\ldots m}L: \bar{P}_{1\ldots m} \bar L& \\ 
        &\rho_{P_{1\ldots m} \bar{P}_{1\ldots m}L\bar{L}} \geq 0& \\
        &\operatorname{tr}_{L}[\rho_{P_{1\ldots m}L} ] = \mathds{1}_{P_{1\ldots m}}/d_{P_{1\ldots m}}&\\
        &\operatorname{tr}_{\bar{P}_{1\ldots m}}[\rho_{\bar{P}_{1\ldots m}L} ] = \mathds{1}_{\bar{L}}/d_{\bar{L}}.&\\
\end{aligned}
\end{equation}
It is now evident that the resulting problem has become invariant under a local symmetry represented by $ S_m $, the symmetric group that permutes $ m $ elements.\footnote{An intriguing observation arises when we consider that the Choi state $ C_{\mathcal{M}}^{\otimes m} $ is, in this particular scenario, identical and independent across the $ m $ subsystems $\mathcal{H}_{P_i}\otimes \mathcal{H}_{\bar{P}_i}$. Even more, our framework remains applicable in the case of an (only) permutation-invariant Choi state that exhibits correlations between the subsystems. This presents an interesting area that may encompass unexplored types of noisy channels $ \mathcal{N}$.} Next, we aim to integrate the permutation symmetry of the channel, as indicated in \eqref{eq:optimization_channel_fidelity_symmetry}, with the results derived from \autoref{thm:hierarchy_outer_approx} (a). To tackle the separability problem in \eqref{eq:optimization_channel_fidelity_symmetry}, we recall that the $ n $-level relaxation presented in \cite{berta2021semidefinite} permutes the systems $ \bar{P}^{(i)}_{1\ldots m}\bar{L}^{(i)} $ across $1 \leq i \leq n $ copies.\\

\begin{figure}
    \centering
\begin{equation}
\begin{aligned}
\begin{matrix}
P_1 & \bar{P}_1 & \bar{P}_1^{(2)} & \ldots & \bar{P}_1^{(n)}\\
P_2 & \bar{P}_2 & \bar{P}_2^{(2)} & \ldots & \bar{P}_2^{(n)}\\
\vdots & \vdots & \ddots & \ldots & \vdots \\
P_{m} & \bar{P}_{m} & \bar{P}_{m}^{(2)}   & \ldots & \bar{P}_{m}^{(n)}\\
L & \bar{L}^{(1)} & \bar{L}^{(2)} & \ldots & \bar{L}^{(n)}
\end{matrix}
\end{aligned}
\end{equation}
\caption{The subsystem structure for the $n$-level relaxation is drawn. The extendibility symmetry enforces a $S_n$-symmetry within the entire $n$ columns counting from the right, i.e.~without the first column from the left. The iid-symmetry is given by a group $S_m$ which permutes the rows of $P$'s in the first two columns from the left.}
    \label{fig:symmetry}
\end{figure}

In order to find a SDP-relaxation of \eqref{eq:optimization_channel_fidelity_symmetry}, we have denoted in \autoref{fig:symmetry} all subsystems, which are of our interest. The symmetry coming from the global symmetry permutes the copied systems of $\bar{P}_{1\ldots m}\bar{L}$ with the copies right of them. It permutes the columns entirely. Copies are denoted with superscript ${}^{(i)}$. The symmetry coming from the $m$-fold product permutes only the rows in the first two columns and without the last row, because it corresponds to the Choi matrix of the identity channel. The $n$-level relaxation of \eqref{eq:optimization_channel_fidelity_symmetry} can be written as
\begin{equation}\label{eq:sdp_revisited}
    \begin{aligned}
        \mathcal{F}(\mathcal{N},d_L) = \max \ &d_P^2 \operatorname{tr}\bigl[C_{\operatorname{id}_L \otimes \mathcal{N}}\,\rho_{LP \bar{L}\bar{P}}\bigr]\\
        \operatorname{s.th.} \ & \rho_{AB} \in \Sigma_{\operatorname{prod}}^n(LP:\bar{L}\bar{P}).
    \end{aligned}
\end{equation}

The action we propose in this section is that, by considering \autoref{fig:symmetry}, the  iid-symmetry is extended to permute entire rows; the isotropic state symmetry acts diagonally on all $L\bar{L}^{(1\ldots n)}$ systems; and the Choi symmetry acts diagonally on all $P_{1\ldots m}\bar{P}_{1\ldots m}^{(1\ldots n)}$ systems.


\subsection{Establishing symmetrized SDPs}

The main idea here is to assume that the $n$-extended Choi matrix remains invariant under the extension of both symmetries to the entire collection of extension systems. To complete the argument, it remains to verify that applying this action across all systems does not conflict with the constraints. The following proposition addresses this point.

\begin{proposition}[Enlarged invariance]\label{prop:enlarged_invariance}
Consider the $n$-level relaxation of \eqref{eq:sdp_revisited} coming from \autoref{thm:hierarchy_outer_approx} (a) as
\begin{equation}\label{eq:enlarge_symmetry_sdp}
        \begin{aligned}
            \operatorname{SDP}_n^\star = &\operatorname{tr}\left[C_{PP} \otimes \Phi_{L\bar{L}} \otimes \mathds{1}_{\bar{P}^{(2\ldots n)}\bar{L}^{(2\ldots n)}} \rho_{P\bar{P}^{(1\ldots n)}L\bar{L}^{(1\ldots n)}}\right]&  \\
            \operatorname{s.th.} \ &\operatorname{tr}[\rho_{P\bar{P}^{(1\ldots n)}L\bar{L}^{(1\ldots n)}}] = 1&\\
            &\operatorname{tr}_{L}[\rho_{P\bar{P}^{(1\ldots n)}L\bar{L}^{(1\ldots n)}}] = \frac{\mathds{1}_{P}}{d_{P}} \otimes \rho_{\bar{P}^{(1\ldots n)}\bar{L}^{(1\ldots n)}} &\\
            &\operatorname{tr}_{\bar{P}^{(n)}}[\rho_{\bar{P}^{(1\ldots n)}\bar{L}^{(1\ldots n)}}] = \rho_{\bar{P}^{(1\ldots n-1)}\bar{L}^{(1\ldots n-1)}} \otimes \frac{\mathds{1}_{\bar{L}^{(n)}}}{d_{\bar{L}^{(n)}}}& \\
            &U_\sigma \rho_{P\bar{P}^{(1\ldots n)}L\bar{L}^{(1\ldots n)}} U_\sigma^{\dagger} = \rho_{P\bar{P}^{(1\ldots n)}L\bar{L}^{(1\ldots n)}} \quad \sigma \in S_n&.
        \end{aligned}
    \end{equation}
    Then, for a solution $\rho_{P\bar{P}^{(1\ldots n)}L\bar{L}^{(1\ldots n)}}^\star$ with value $\operatorname{SDP}_n^\star$ there exists a solution
    \begin{enumerate}
        \item[(a)] $\tilde{\rho}_{P\bar{P}^{(1\ldots n)}L\bar{L}^{(1\ldots n)}}^\star$ invariant under the iid-symmetry $S_m$, which acts diagonally on $P_{1\ldots m}\bar{P}_{1\ldots m}^{(1\ldots n)}$, and it satisfies the constraints whenever $C_{P_{1\ldots m}\bar{P}_{1\ldots m}}$ is invariant under the iid-symmetry defined in \autoref{subsec:types_of_symmetries}~(c).
        \item[(b)] $\tilde{\rho}_{P\bar{P}^{(1\ldots n)}L\bar{L}^{(1\ldots n)}}^\star$ invariant under the Choi-symmetry $\mathcal{V}$, which acts diagonally on $P_{1\ldots m}\bar{P}_{1\ldots m}^{(1\ldots n)}$, and it satisfies the constraints whenever $C_{P_{1\ldots m}\bar{P}_{1\ldots m}}$ is invariant under the Choi-symmetry defined in \autoref{subsec:types_of_symmetries}~(b).
        \item[(c)]$\tilde{\rho}_{P\bar{P}^{(1\ldots n)}L\bar{L}^{(1\ldots n)}}^\star$ invariant under the isotropic state symmetry $\mathcal{U}(\mathbb{C}^2)$ defined in \autoref{subsec:types_of_symmetries}~(d), which acts diagonally on $L\bar{L}^{(1\ldots n)}$.
    \end{enumerate}
    Moreover, all the symmetries together build the group $S_m \times S_n \times \mathcal{U}\times \mathcal{V}$, whereby $S_m$ is the iid-symmetry, $S_n$ the extendibility symmetry, $\mathcal{U}(\mathbb{C}^2)$ is the unitary symmetry acting diagonal on all $L\bar{L}^{(1\ldots n)}$ systems, and $\mathcal{V}$ is the Choi symmetry acing diagonal on all $P_{1\ldots m}\bar{P}_{1\ldots m}^{(1\ldots n)}$ systems.
\end{proposition}

\begin{proof}
    We make the calculation for one case in detail and argue afterwards that it works similarly in all three cases. We consider the symmetric group $S_m$ acting diagonal on all $P$ systems and denote its representation as $\mathcal{W}$. For sake of generality we write the group average with respect to a measure $d\mu(W)$ as an integral even though for finite groups it is the counting measure with normalization constant, but for further calculations involving the unitary group we aim to stay as general as possible. So, consider
    \begin{align*}
        \tilde{\rho}_{P\bar{P}^{(1\ldots n)}L\bar{L}^{(1\ldots n)}}^\star \coloneqq \int_G d\mu(W) W^{\otimes n+1} \rho_{P\bar{P}^{(1\ldots n)}L\bar{L}^{(1\ldots n)}}^\star(W^{\otimes n})^{\dagger}.
    \end{align*}
    We have to show that $\tilde{\rho}_{P\bar{P}^{(1\ldots n)}L\bar{L}^{(1\ldots n)}}^\star$ satisfies the constraints if $\rho_{P\bar{P}^{(1\ldots n)}L\bar{L}^{(1\ldots n)}}^\star$ does. We start with the left hand side of the first constraint
    \begin{equation}\label{eq:large_invariance_prop1}
    \begin{aligned}
        \operatorname{tr}_P[\tilde{\rho}_{P\bar{P}^{(1\ldots n)}L\bar{L}^{(1\ldots n)}}^\star] &= \int_G d\mu(W) \operatorname{tr}_P[W^{\otimes n+1} \rho_{P\bar{P}^{(1\ldots n)}L\bar{L}^{(1\ldots n)}}^\star(W^{\otimes n+1})^{\dagger}] \\
        &=\int_G d\mu(W) W^{\otimes n} \operatorname{tr}_P[\rho_{P\bar{P}^{(1\ldots n)}L\bar{L}^{(1\ldots n)}}^\star](W^{\otimes n})^{\dagger} \\
        &= \int_G d\mu(W) W^{\otimes n}  \frac{\mathds{1}_{L}}{d_{L}} \otimes \rho_{\bar{P}^{(1\ldots n)}\bar{L}^{(1\ldots n)}}^\star (W^{\otimes n})^{\dagger} \\
        &= \frac{\mathds{1}_{L}}{d_{L}} \otimes \int_G d\mu(W) W^{\otimes n}   \rho_{\bar{P}^{(1\ldots n)}\bar{L}^{(1\ldots n)}}^\star (W^{\otimes n})^{\dagger}. 
    \end{aligned}
    \end{equation}
    However, if we trace $P$ and $L$ out in we get 
    \begin{equation}\label{eq:large_invariance_prop2}
        \begin{aligned}
            \operatorname{tr}_{PL}[\tilde{\rho}_{P\bar{P}^{(1\ldots n)}L\bar{L}^{(1\ldots n)}}^\star] &= \int_G d\mu(W) \operatorname{tr}_{PL}[W^{\otimes n+1} \rho_{P\bar{P}^{(1\ldots n)}L\bar{L}^{(1\ldots n)}}^\star(W^{\otimes n+1})^{\dagger}] \\
        &=\int_G d\mu(W) W^{\otimes n} \operatorname{tr}_{PL}[\rho_{P\bar{P}^{(1\ldots n)}L\bar{L}^{(1\ldots n)}}^\star](W^{\otimes n})^{\dagger} \\
        &= \int_G d\mu(W) W^{\otimes n}  \rho_{\bar{P}^{(1\ldots n)}\bar{L}^{(1\ldots n)}}^\star (W^{\otimes n})^{\dagger} \\
        &= \int_G d\mu(W) W^{\otimes n}   \rho_{\bar{P}^{(1\ldots n)}\bar{L}^{(1\ldots n)}}^\star (W^{\otimes n})^{\dagger}\\
        &= \tilde{\rho}_{\bar{P}^{(1\ldots n)}\bar{L}^{(1\ldots n)}}^\star.
        \end{aligned}
    \end{equation}
    Thus, we conclude that $\tilde{\rho}_{P\bar{P}^{(1\ldots n)}L\bar{L}^{(1\ldots n)}}^\star$ satisfies the constraints. In particular, if $\rho_{P\bar{P}^{(1\ldots n)}L\bar{L}^{(1\ldots n)}}^\star$ is optimal for the SDP with value $c^\star$ and the objective is invariant under $W\otimes W$, we find
    \begin{align*}
        c^\star  &= \operatorname{tr}\left[C_{PP}\otimes \Phi_{L\bar{L}}\otimes \mathds{1}_{\bar{P}^{(1\ldots n-1)}\bar{L}^{(1\ldots n-1)}} \rho_{P\bar{P}^{(1\ldots n)}L\bar{L}^{(1\ldots n)}}^\star\right] \\
        &= \operatorname{tr}\left[\int_G d\mu(W) (W^{\otimes n+1})^\dagger C_{PP}\otimes \Phi_{L\bar{L}}\otimes \mathds{1}_{\bar{P}^{(1\ldots n-1)}\bar{L}^{(1\ldots n-1)}} W^{\otimes n+1} \rho_{P\bar{P}^{(1\ldots n)}L\bar{L}^{(1\ldots n)}}^\star\right] \\
        &= \operatorname{tr}\left[  C_{PP}\otimes \Phi_{L\bar{L}}\otimes \mathds{1}_{\bar{P}^{(1\ldots n-1)}\bar{L}^{(1\ldots n-1)}} \int_G d\mu(W) W^{\otimes n+1} \rho_{P\bar{P}^{(1\ldots n)}L\bar{L}^{(1\ldots n)}}^\star(W^{\otimes n+1})^\dagger \right] \\
        &= \operatorname{tr}\left[  C_{PP}\otimes \Phi_{L\bar{L}}\otimes \mathds{1}_{\bar{P}^{(1\ldots n-1)}\bar{L}^{(1\ldots n-1)} }\tilde{\rho}_{P\bar{P}^{(1\ldots n)}L\bar{L}^{(1\ldots n)}}^\star\right].
    \end{align*}
    Consequently $\tilde{\rho}_{P\bar{P}^{(1\ldots n)}L\bar{L}^{(1\ldots n)}}^\star$ has the same value $c^\star$ as $\rho_{P\bar{P}^{(1\ldots n)}L\bar{L}^{(1\ldots n)}}^\star$. Moreover, we conclude that the calculations for the other symmetries work exactly similar. The important ingredient is always that the symmetries act uniformly on the $P$ and $L$ system, such that we can use cyclicity of the partial trace.  
\end{proof}

A particularly interesting example for applications is the result for the iid-symmetry combined with the isotropic state symmetry and the extension symmetry\,---\,as this is a standard setting in coding tasks. 


\subsection{Solving the SDPs}

As we can now apply \autoref{prop:enlarged_invariance}, we are in the standard framework for symmetry reduction of SDPs. In the following we consider a group $\mathcal{W}$ consisting out of some combinations of symmetries from \autoref{subsec:types_of_symmetries}. In the standard framework \cite{Klerk2007ReductionOS}, it is assumed to have access to a basis $\mathcal{B}$ of the commutant $\mathcal{M}^\prime_{\mathcal{W}}(\mathcal{H}_{P \bar{P}^{(1\ldots n)}L \bar{L}^{(1\ldots n)}})$. Then, we can rewrite each $\rho \in \mathcal{M}^\prime_{\mathcal{W}}(\mathcal{H}_{P \bar{P}^{(1\ldots n)}L \bar{L}^{(1\ldots n)}})$ in 
\begin{align}\label{eq:objective_decomposition_in_basis}
    \rho = \sum_{B \in \mathcal{B}} x_B B
\end{align}
and optimize over the coefficients $x_B$, whereby the size of the resulting SDPs scales in the optimal case polynomial with $\vert B \vert$, the size of the basis. An important tool, which is standard in the literature is that for a given group $\mathcal{W}$, the map 
\begin{align}\label{eq:star_homomophism}
   \psi_1:\mathcal{M}^\prime_{\mathcal{W}}(\mathcal{H}_{P \bar{P}^{(1\ldots n)}L \bar{L}^{(1\ldots n)}}) &\to \bigoplus_{\lambda \in \operatorname{Irr}(\mathcal{W})}\mathbb{C}^{m_\lambda \times m_\lambda}
\end{align}
is a $\star$-homomorphism, which in particular means that $\psi_1$ and $\psi_1^{-1}$ are both positive. Usually this map is constructed via \emph{representative sets}, i.e. sets with one representative per irreducible representation in the full decomposition of the $\mathbb{C} \mathcal{W}$-module. Additionally, we need to handle the constraints of the AQEC problem appropriately, which can be done similarly to \cite[Lem. 3.3., Thm.~3.5]{chee2024efficient} in the following lemma.

\begin{lemma}\label{lem:making_the_constraints_small}
    Assume we are in the setup of \autoref{prop:enlarged_invariance} and have either of the symmetries there. Then, we have that the equations of the constraints 
    \begin{align}\label{eq:constraints_in_lemma}
        \operatorname{tr}_{P}[\rho_{P\bar{P}^{(1\ldots n)}L\bar{L}^{(1\ldots n)}}] &= \frac{\mathds{1}_{L}}{d_{L}} \otimes \rho_{\bar{P}^{(1\ldots n)}\bar{L}^{(1\ldots n)}}\\
        \operatorname{tr}_{L}[\rho_{\bar{P}^{(1\ldots n)}\bar{L}^{(1\ldots n)}}] &= \rho_{\bar{P}^{(1\ldots n-1)}\bar{L}^{(1\ldots n-1)}} \otimes \frac{\mathds{1}_{P}}{d_P}
    \end{align}
    can be handled as equality constraints in a certain sub-algebra of $\mathcal{B}(\mathcal{H})$, depending on the symmetries from \autoref{prop:enlarged_invariance} as
    \begin{enumerate}
        \item[(a)] $\tilde{\rho}_{P\bar{P}^{(1\ldots n)}L\bar{L}^{(1\ldots n)}}^\star$ invariant under the action of the iid-symmetry $S_m$ acting diagonal on $P\cong \bigotimes_{i=1}^m P_i$: then the first equation in \eqref{eq:constraints_in_lemma} can be solved in the $n$-fold diagonal action of $S_m$, that is $\mathcal{M}^\prime_{S_m}(\mathcal{H}_L\otimes \mathcal{H}_{\bar{P}^{(1\ldots n)}\bar{L}^{(1\ldots n)}})$, and the second equation in the $n$-fold diagonal action of $S_m$ over the $\bar{P}_{1\ldots m}$ systems, that is $\mathcal{M}^\prime_{S_m}(\mathcal{H}_{\bar{P}^{(1\ldots n)}\bar{L}^{(1\ldots n-1)}})$.
        \item[(b)] If $\tilde{\rho}_{P\bar{P}^{(1\ldots n)}L\bar{L}^{(1\ldots n)}}^\star$ is invariant under the action of a possible Choi symmetry $\mathcal{V}$ acting diagonal on all $P_{1\ldots m}\bar{P}^{(i)}_{1\ldots m}$, $1\leq i \leq n$, systems, then the first equation in \eqref{eq:constraints_in_lemma} can be solved in the $n\cdot m$-fold diagonal action of the Choi symmetry, that is $\mathcal{M}^\prime_{\mathcal{W}}(\mathcal{H}_L\otimes \mathcal{H}_{\bar{P}_{1\ldots m}^{(1\ldots n)}\bar{L}^{(1\ldots n)}})$, and the second equation in the $n\cdot m$-fold diagonal action of the Choi symmetry, that is $\mathcal{M}^\prime_{\mathcal{W}}(\mathcal{H}_{\bar{P}_{1\ldots m}^{(1\ldots n)}\bar{L}^{(1\ldots n-1)}})$.
        \item[(c)] $\tilde{\rho}_{P\bar{P}^{(1\ldots n)}L\bar{L}^{(1\ldots n)}}^\star$ invariant under the unitary-symmetry $U^{\otimes n+1}$ acting diagonal on the $L$-systems, then the first equation in \eqref{eq:constraints_in_lemma} can be solved in the $(n+1)$-fold diagonal action of $\mathcal{U}$, that is $\mathcal{M}^\prime_{\mathcal{U}}(\mathcal{H}_L\otimes \mathcal{H}_{\bar{P}^{(1\ldots n)}\bar{L}^{(1\ldots n)}})$, and the second equation in the $(n-1)$-fold diagonal action of $U$, that is $\mathcal{M}^\prime_{\mathcal{U}}(\mathcal{H}_{\bar{P}^{(1\ldots n)}\bar{L}^{(1\ldots n-1)}})$.
    \end{enumerate}
    Due to the fact that the symmetries commute, we can combine either of the symmetries (a)-(c) and get the direct product of groups, such that the constraints are satisfied. 
\end{lemma}

\begin{proof}
    For the proof we refer to \eqref{eq:large_invariance_prop1} and \eqref{eq:large_invariance_prop2}, where it is easy to see that the reduced state possess the rest of underlying symmetry.     
\end{proof}

\autoref{lem:making_the_constraints_small} states that constraints in \eqref{eq:constraints_in_lemma} can be seen as constraints in a sub-algebra and thus can be handled in the sub-algebra. Particularly, a combined symmetry out of (a)-(c) induces a corresponding $\star$-homomorphism similar to \eqref{eq:star_homomophism}. We summarize the steps in the following theorem.

\begin{theorem}\label{thm:small_SDPs}
    Consider the SDP in \eqref{eq:enlarge_symmetry_sdp}. Then, the objective variable can be rewritten with a matrix basis $\mathcal{B}$ of $\mathcal{M}^\prime_{\mathcal{W}}(\mathcal{H}_{P\bar{P}^{(1\ldots n)}L\bar{L}^{(1\ldots n)}})$ corresponding to a $\star$-isomorphism $\psi_1$ and $\star$-isomorphisms $\psi_2$ and $\psi_3$ corresponding to the constraints in \eqref{eq:constraints_in_lemma} as
    \begin{equation}
    \begin{aligned}
        \psi_1:\mathcal{M}^\prime_{\mathcal{W}}(\mathcal{H}_{P \bar{P}^{(1\ldots n)}L \bar{L}^{(1\ldots n)}}) &\to \bigoplus_{\lambda \in \operatorname{Irr}(\mathcal{W})}\mathbb{C}^{m_\lambda \times m_\lambda} \\
        \psi_2:\mathcal{M}^\prime_{\mathcal{W}}(\mathcal{H}_{P\bar{P}^{(1\ldots n)} \bar{L}^{(1\ldots n)}}) &\to \bigoplus_{\lambda \in \operatorname{Irr}(\mathcal{W})}\mathbb{C}^{n_\lambda \times n_\lambda}  \\
        \psi_3: \mathcal{M}^\prime_{\mathcal{W}}(\mathcal{H}_{\bar{P}^{(1\ldots n-1)} \bar{L}^{(1\ldots n)}}) &\to  \bigoplus_{\lambda \in \operatorname{Irr}(\mathcal{W})} \mathbb{C}^{f_\lambda \times f_\lambda}.
    \end{aligned}
    \end{equation} 
    Then we can rewrite the objective variable with respect to the basis $\mathcal{B}$ as done in \eqref{eq:objective_decomposition_in_basis} such that the AQEC problem can be equivalently written as
     \begin{equation}\label{eq:small_problem}
    \begin{aligned}
        \sup_{x_B, B \in \mathcal{B}} &\sum_B x_B \operatorname{tr}\Big[C_{PP}\otimes \Phi_{L\bar{L}}\operatorname{tr}_{P^{n}_2 L^{n}_2}[B]\big]& \\
        \operatorname{s.th.} \ &\sum_B x_B \psi^{(\lambda)}_1(B) \geq 0, \quad \lambda \in \operatorname{Irr}(\mathcal{W})& \\
        & \sum_B x_B \operatorname{tr}[B] = 1 &\\
        & \sum_B x_B \psi_2^{(\lambda)} (\operatorname{tr}_P[B]) = \sum_B x_B \psi_2^{(\lambda)} \left(\frac{\mathds{1}_L}{d_L} \otimes \operatorname{tr}_{LP}[B]\right), \quad \lambda \in \operatorname{Irr}(\mathcal{W})& \\
        & \sum_B x_b \psi_3^{(\lambda)} (\operatorname{tr}_{PLL}[B]) = \sum_B x_B \psi_3^{(\lambda)} \left( \operatorname{tr}_{PLPL}[B] \otimes \frac{\mathds{1}_P}{d_P}\right), \quad \lambda \in \operatorname{Irr}(\mathcal{W}).& 
    \end{aligned}
    \end{equation}
    The resources needed in order to solve the resulting SDP \eqref{eq:small_problem} can be summarized as
    \begin{itemize}
        \item $\vert \mathcal{B}\vert$ many optimization variables
        \item for each irreducible representation $\lambda \in \operatorname{Irr}(\mathcal{W})$ a positive semi-definite matrix variable of size $m_\lambda\times m_\lambda$ 
        \item for each irreducible representation $\lambda \in \operatorname{Irr}(\mathcal{W})$ $n_\lambda^2$ many equality constraints
        \item for each irreducible representation $\lambda \in \operatorname{Irr}(\mathcal{W})$ $f_\lambda^2$ many equality constraints.
    \end{itemize}
\end{theorem}

\begin{proof}
    If we start with the AQEC problem as stated in \eqref{eq:sdp_revisited}, we first assume a subset of the three possible symmetries: iid, isotropic state and local channel symmetry. 
    Due to the fact that in our representation in \autoref{prop:enlarged_invariance} all of them commute, we can just choose a problem dependent subset of them. Then, we form the global symmetry group, which is the direct product of all of them with the symmetric group $S_n$ as the extension symmetry. Again, since all possible additional symmetries commute with $S_n$, the direct product structure is exactly what comes up in this case. This group is then called $\mathcal{W}$. Now, again \autoref{prop:enlarged_invariance} states that we can restrict the optimization to the commutant 
    \begin{align}
    \mathcal{M}^\prime_{\mathcal{W}}(\mathcal{H}_{P\bar{P}^{(1\ldots n)}L\bar{L}^{(1\ldots n)}})    
    \end{align}
    and get the same optimal value. Assuming now access to a basis $\mathcal{B}$ of $\mathcal{M}^\prime_{\mathcal{W}}(\mathcal{H}_{P\bar{P}^{(1\ldots n)}L\bar{L}^{(1\ldots n)}})$, yields that we can express each element as a linear combination of those basis elements. From the fact that the maps $\psi_i$, $1\leq i \leq 3$ are $\star$-isomorphisms and thus indeed positive (and the inverses are positive as well), we can express positivity in terms of positivity in the direct sum decomposition, which justifies the first constraints. The second constraint is that we want to have trace normalization at the end. From \autoref{lem:making_the_constraints_small} we conclude that the constraints live in corresponding endomorphism algebras corresponding to $\mathcal{W}$, respectively its restriction to a subgroup (for the particular case of $S_n$ to $S_{n-1}$), such that again enforcing them in the large space is due to the fact that $\psi_2$ and $\psi_3$ are $\star$-isomorphisms equivalent to enforce them in the direct sum spaces. 

    We conclude that we need $\vert \mathcal{B}\vert$ many optimization variables for the basis decomposition, positive semi-definite matrices corresponding to the positivity constraints in the direct sum decomposition corresponding to $\psi_1$ and for each irreducible $\lambda$ occurring in $\psi_2$ and $\psi_3$ $n_\lambda^2$ respectively $f_\lambda^2$ equality constraints, because those equations are full matrix equalities, which have to be enforced entry-wise.  
\end{proof}


\subsection{Details about the optimization program}

\autoref{thm:small_SDPs} gives a symmetry reduced formulation for AQEC with access to different symmetries. For actual calculations, several challenges need to be resolved. First of all, the theorem assumes access to a basis $\mathcal{B}$ of the endomorphism space. For theoretical proofs about efficiency, a standard basis can be chosen in terms of the orbit construction by \cite{Klerk2007ReductionOS}. This basis has the particularly nice advantage that for certain cases of $\mathcal{W}$, such as, e.g., the extension symmetry \cite{chee2024efficient}, a polynomial scaling method is known to pre-compute all the relevant data in \autoref{thm:small_SDPs}. It has to be stated that the resulting SDP in \autoref{thm:small_SDPs} often scales polynomial in, e.g., the dimensions of $\mathcal{H}_P$ and $\mathcal{H}_L$, but to pre-compute the program, a priori needs many operations in the exponentially large spaces in those dimensions.


\subsubsection{Advantages of a Hilbert Schmidt orthonormal basis in practice}

As mentioned in the previous paragraph, for theoretical results, one often uses the orbit basis, because there are well known results for the computation of the entries in \autoref{thm:small_SDPs} in polynomial time in the dimension of $\mathcal{H}_P$ and $\mathcal{H}_L$. However, if we consider the resources generally needed in order to compute \autoref{thm:small_SDPs} in practice, we observe that the positivity constraints and the number of summands for the constraints may not scale optimally for an implementation. For this reason, we consider a Hilbert-Schmidt orthogonal basis 
\begin{align}
    \{(B^\lambda_{ij})_{i,j = 1}^{m_\lambda}\}_\lambda \subset \mathcal{M}^\prime_{\mathcal{W}}(\mathcal{H}_{P \bar{P}^{(1\ldots n)}L \bar{L}^{(1\ldots n)}}),
\end{align}
such that 
\begin{align}
    \operatorname{tr}[B_{ij}^\lambda] = \delta_{ij} \quad \text{for all} \ 1\leq i,j \leq m_\lambda, \ \lambda \in \operatorname{Irr}(\mathcal{W}).
\end{align}
Assume we have a unitary $U$, which decomposes $\mathcal{H}_{P \bar{P}^{(1\ldots n)}L \bar{L}^{(1\ldots n)}}$ into the direct sum if isotypic components and in particular 
\begin{align}\label{eq:decomposition_operators}
    \mathcal{B}(\mathcal{H}_{P \bar{P}^{(1\ldots n)}L \bar{L}^{(1\ldots n)}}) \equiv U \bigoplus_{\lambda \in \operatorname{Irr}(\mathcal{W})} \mathbb{C}^{m_\lambda \times m_\lambda} \otimes \mathbb{C}^{d_\lambda \times d_\lambda} U^\dagger,
\end{align}
where the equation should be read as "up to vector space isomorphic identifications". Then we can extract a subset $U_{\mathcal{B}}^\lambda$ of $U$ considered as a basis for $\mathcal{H}_{P \bar{P}^{(1\ldots n)}L \bar{L}^{(1\ldots n)}}$ for each irreducible representation and furthermore representative sets $U^\lambda_{\mathcal{R}}$ which concretely formulate the maps $\psi_1^{(\lambda)}$.

\begin{lemma}\label{lem:HSbasis}
Consider the decomposition in \eqref{eq:decomposition_operators} of the operator space according to the isotypic components. For each irreducible representation $\lambda$, let $U_{\mathcal{B}}^\lambda$ be a unitary whose columns form a basis of the $\lambda$-sector and let $\{E_{ij}\}_{i,j=1}^{m_\lambda}$ be the standard (elementary) basis of $\mathbb{C}^{m_\lambda\times m_\lambda}$ and define
\begin{align}
    B_{ij}^{\lambda} \coloneqq U_{\mathcal{B}}^\lambda \Bigl(E_{ij}\otimes \frac{\mathds{1}_{d_\lambda}}{d_\lambda}\Bigr)
    (U_{\mathcal{B}}^\lambda)^\dagger.
\end{align}
Then, the collection $\{B_{ij}^{\lambda}\}_{\lambda,i,j}$ forms a Hilbert--Schmidt orthogonal basis of $\mathcal{M}^\prime_{\mathcal{W}}(\mathcal{H}_{P \bar{P}^{(1\ldots n)}L \bar{L}^{(1\ldots n)}})$ and for every $\lambda$ and indices $i,j$ one has
\begin{align}
    \operatorname{tr}\Bigl[B_{ij}^{\lambda}\Bigr]=\delta_{ij},
\end{align}
so that in particular the diagonal elements have trace one.
\end{lemma}

\begin{proof}
Since the Hilbert Schmidt inner product is unitarily invariant, it suffices to work in the block-diagonal picture. 
Thus, for any $\lambda,\mu$ and any indices $i,j,k,l$,
\begin{align}
\langle B_{ij}^{\lambda},B_{kl}^{\mu}\rangle
&=\operatorname{tr}\Bigl[(B_{ij}^{\lambda})^\dagger B_{kl}^{\mu}\Bigr]
=\operatorname{tr}\Bigl[\Bigl(E_{ij}\otimes \frac{\mathds{1}_{d_\lambda}}{d_\lambda}\Bigr)^\dagger \Bigl(E_{kl}\otimes \frac{\mathds{1}_{d_\mu}}{d_\mu}\Bigr)\Bigr]\\
&=\operatorname{tr}\Bigl[\bigl(E_{ij}^\dagger\otimes \frac{\mathds{1}_{d_\lambda}}{d_\lambda}\bigr)
\bigl(E_{kl}\otimes \frac{\mathds{1}_{d_\mu}}{d_\mu}\bigr)\Bigr].
\end{align}
Since $E_{ij}^\dagger = E_{ji}$ and using the property of the tensor product we have
\begin{align}
\langle B_{ij}^{\lambda},B_{kl}^{\mu}\rangle
&=\operatorname{tr}\Bigl[(E_{ji}E_{kl})\otimes \Bigl(\frac{\mathds{1}_{d_\lambda}}{d_\lambda}\cdot\frac{\mathds{1}_{d_\mu}}{d_\mu}\Bigr)\Bigr].
\end{align}
Because the different isotypic components are mutually orthogonal (i.e., the block corresponding to $\lambda$ is orthogonal to that corresponding to $\mu$ for $\lambda\neq\mu$), it follows that
\begin{align}
\langle B_{ij}^{\lambda},B_{kl}^{\mu}\rangle
=\delta_{\lambda\mu}\; \operatorname{tr}\Bigl[(E_{ji}E_{kl})\otimes \frac{\mathds{1}_{d_\lambda}}{d_\lambda^2}\Bigr].
\end{align}
Recall that the elementary matrices satisfy
\begin{align}
E_{ji}E_{kl} = \delta_{i,k}\, E_{jl},\qquad \text{with}\quad \operatorname{tr}(E_{jl})=\delta_{j,l}.
\end{align}
Moreover, we have
\begin{align}
\operatorname{tr}\Bigl[\frac{\mathds{1}_{d_\lambda}}{d_\lambda^2}\Bigr]=\frac{d_\lambda}{d_\lambda^2}=\frac{1}{d_\lambda}.
\end{align}
Thus, we obtain
\begin{align}
\langle B_{ij}^{\lambda},B_{kl}^{\mu}\rangle
=\delta_{\lambda\mu}\,\delta_{i,k}\,\delta_{j,l}\,\frac{1}{d_\lambda}.
\end{align}
This shows that the family $\{B_{ij}^\lambda\}_{\lambda,i,j}$ is Hilbert--Schmidt orthogonal (up to the normalization factor $1/d_\lambda$; one may always rescale by $\sqrt{d_\lambda}$ to obtain an orthonormal basis if desired).

Next, we verify the trace property. Since the trace is also unitarily invariant, we have
\begin{align}
\operatorname{tr}\Bigl[B_{ij}^{\lambda}\Bigr]
&=\operatorname{tr}\Bigl[E_{ij}\otimes \frac{\mathds{1}_{d_\lambda}}{d_\lambda}\Bigr]
=\operatorname{tr}[E_{ij}]\,\operatorname{tr}\Bigl[\frac{\mathds{1}_{d_\lambda}}{d_\lambda}\Bigr] \\
&=\delta_{ij}\cdot 1
=\delta_{ij}.
\end{align}
In particular, for the diagonal elements ($i=j$) we have $\operatorname{tr}[B_{ii}^{\lambda}]=1$.
\end{proof}

If we now apply the maps $\psi_1^{(\lambda)}$ to a basis element $B_{ij}^\mu$, we get
\begin{align}\label{eq:basis_elements}
    \psi_1^{(\lambda)}(B_{ij}^\mu) = E_{ij}\delta_{\lambda \mu}.
\end{align}
That is, the full set of positivity constraints can just be replaced by positive semi-definite variables in the SDP of sizes $m_\lambda \times m_\lambda$. As a second advantage of such a Hilbert-Schmidt orthonormal basis, we extract that when reducing from $S_n$ to $S_{n-1}$ in the second constraint, the elements $B_{ij}^\lambda$, $\lambda \vdash n$ have just overlap with those partition $\mu \vdash n-1$ which become direct summands of the restricted $S_{n-1}$-module (see, e.g., \cite{Sagan2001} for Branching rules of the symmetric group). Thus, we reduce the amount of summands in the constraints. Moreover, due to the fact that our diagonal actions can be viewed as duals of a Schur-Weyl duality pair exactly something similar happens for those symmetries in case we have chosen a basis with respect to the irreducibles.  


\subsection{Small-scale examples}\label{subsec:commutative_relax_examples}

This section is dedicated to a few abstract cases and an example in AQEC. To get connected with standard literature in the domain of error correction, we assume from now on to encode one qubit into $m$ qubits with $n$-extensions in the hierarchy. Concretely, we consider 
\begin{align}
    \mathcal{S}(\mathcal{H}_L) \to \mathcal{S}(\mathcal{H}_P) \to \mathcal{S}(\mathcal{H}_P) \to \mathcal{S}(\mathcal{H}_L)   
\end{align}
where $\operatorname{dim} \mathcal{H}_L = 2$ and $\mathcal{H}_P = 2^m$.


\subsubsection{Iid-symmetry, isotropic state symmetry and extension symmetry}
\label{subsub:iid_werner_ext}

This section considers the situation of \autoref{fig:symmetry} with an $S_m$ action diagonal on all $P_j^{(i)}$-systems, $1\leq j \leq m$ and $1\leq i \leq n$, a unitary action diagonal on all $L$ systems and a permutation action of the $S_n$ on the last $n$ tensor factors or, in other words, on the extension due to the hierarchy in \autoref{thm:hierarchy_outer_approx}. The abstract strategy for working with $S_m \times S_n \times \mathcal{U}(2)$-modules is always to consider first the $S_m \times S_n$-module $P_{1\ldots m}\bar{P}_{1\ldots m}^{(1\ldots n)}$ and the $\mathcal{U}(2) \times S_n$-module $L\bar{L}^{(1\ldots n)}$. If we have a structural formula for both of them, we can tensor the modules and extract a decomposition into $S_m \times S_n \times \mathcal{U}(2)$ modules. 

So, starting with $\mathcal{H}_{L\bar{L}^{(1\ldots n)}}$, we can simply adopt the calculation coming from Schur-Weyl duality and branching rules. Considered as a $S_{n+1} \times \mathcal{U}(2)$ action, we have
\begin{align}\label{eq:decomposition_L}
    \mathcal{H}_{L\bar{L}^{(1\ldots n)}} \cong_{S_{n+1}\times \mathcal{U}(2)} \bigoplus_{\lambda \vdash n+1} E^\lambda \otimes S^\lambda. 
\end{align}
Now, we just restrict the $S_{n+1}$ to a subgroup $S_n$, which decomposes the $S^\lambda$ modules according to the branching rule, which can be stated as follows
\begin{align}
    S^\lambda \downarrow_{S_{n}} \cong_{S_n} \bigoplus_{A \in \operatorname{Rem}(\lambda)} S^{\lambda_A}.
\end{align}
Here, we follow the standard notations in the representation theory of the symmetric group: a box $A \in Y(\lambda)$ is a removable box of $\lambda$ if $Y(\lambda) \backslash \{A\} = Y(\mu)$ for a partition $\mu \vdash n$ and we write $\lambda_A$ for this partition $\mu$. The set of removable boxes is then denoted as $\operatorname{Rem}(\lambda)$. Inserting this equation just straight ahead into \eqref{eq:decomposition_L}, we get
\begin{align}
    \mathcal{H}_{L\bar{L}^{(1\ldots n)}} \cong_{S_{n}\times \mathcal{U}(2)} \bigoplus_{\lambda \vdash n+1} E^\lambda \otimes \bigoplus_{A \in \operatorname{Rem}(\lambda)} S^{\lambda_A} = \bigoplus_{\lambda \vdash n+1}\bigoplus_{A \in \operatorname{Rem}(\lambda)}  E^\lambda \otimes S^{\lambda_A}.
\end{align}
In particular, we see that the decomposition is multiplicity free, such that all occurring variables would be LP-variables so far.

For the diagonal permutation invariant $P_{1\ldots m}\bar{P}_{1\ldots m}^{(1\ldots n)}$ part with the extension symmetry permuting the last $n$ tensor factors of $\bar{P}_{1\ldots m}^{(1\ldots n)}$ with $\operatorname{dim} P_{j}^{(i)} = 2$, we give the concrete example $n=2$ and $m=3$. The symmetry group then becomes $S_2\times S_3$. For this group, all representation theory is known with the character table

\begin{align}\label{eq:character_table_S3S2}
\begin{array}{c|cccccc|c}
 S_3\times S_2& (e,e) & (e,\sigma) & (\tau,e) & (\tau,\sigma) & (\rho,e) & (\rho,\sigma) & \text{multiplicities}\\
\hline
(1_+,1^+) & 1 & 1 & 1 & 1 & 1 & 1 & 74\\
(1_+,1^-) & 1 & -1 & 1 & -1 & 1 & -1 & 26 \\
(1_-,1^+) & 1 & 1 & -1 & -1 & 1 & 1 & 30\\
(1_-,1^-) & 1 & -1 & -1 & 1 & 1 & -1 & 46\\
(2,1^+)   & 2 & 2 &  0 &  0 & -1 & -1 & 74\\
(2,1^-)   & 2 & -2 &  0 &  0 & -1 &  1 & 94. \\
\end{array}
\end{align}

In \eqref{eq:character_table_S3S2} $\tau$ is a class function representative for the transpositions in $S_3$, similarly $\sigma$ is a representative for the conjugacy classes of transpositions in $S_2$. The variable $\rho$ represents a $3$-cycle in $S_3$. The notation $1_+,1_-$ respectively $1^+,1^-$ represents the trivial and the sign character and $2$ represents the two-dimensional irreducible representation of $S_3$. Moreover, it is straightforward to calculate the multiplicities for this decomposition with character theory. As such, we find
\begin{align}
    74 S^{(3)} \otimes S^{(2)} &\oplus 26 S^{(3)} \otimes S^{(1,1)} \oplus 30 S^{(1,1,1)} \otimes S^{(2)} \oplus 46 S^{(1,1,1)} \otimes S^{(1,1)} \oplus 74 S^{(2,1)}\otimes S^{(2)}\nonumber\\
    &\oplus 94 S^{(2,1)} \otimes S^{(1,1)}.
\end{align}
In particular, we can decompose for the particular case of $n=3$ the decomposition of $\mathcal{H}_{L\bar{L}^{(1,2)}} \cong \operatorname{Sym}^3(L)\otimes S^{(2)} \bigoplus \textbf{2}_{\operatorname{std}} \otimes S^{(2)} \bigoplus\textbf{2}_{\operatorname{std}} \otimes S^{(1,1)}$, where we mean by $\operatorname{Sym}^3(L)$ the symmetric subspace and by $\textbf{2}_{\operatorname{std}}$ the $2$-dimensional irreducible of $L\bar{L}^{(1,2)}$ under diagonal $\mathcal{U}(2)$ action. Now, we can form the tensor product of those modules and get $12$ blocks of sizes

\begin{align}
\begin{array}{|c|c|c|c|}
\hline
\text{multiplicities} & S_3\text{-irreps} & \mathcal{U}(2)\text{-irreps} & S_2\text{-irreps} \\[4mm]
\hline
74  & S^{(3)}      & \operatorname{Sym}^3(L)     & S^{(2)}   \\[2mm]
26  & S^{(3)}      & \operatorname{Sym}^3(L)     & S^{(1,1)} \\[2mm]
30  & S^{(1,1,1)}  & \operatorname{Sym}^3(L)     & S^{(2)}   \\[2mm]
46  & S^{(1,1,1)}  & \operatorname{Sym}^3(L)     & S^{(1,1)} \\[2mm]
74  & S^{(2,1)}    & \operatorname{Sym}^3(L)     & S^{(2)}   \\[2mm]
94  & S^{(2,1)}    & \operatorname{Sym}^3(L)     & S^{(1,1)} \\[2mm]
120 & S^{(3)}      & \mathbf{2}_{\operatorname{std}} & S^{(2)}   \\[2mm]
56 & S^{(3)}      & \mathbf{2}_{\operatorname{std}} & S^{(1,1)} \\[2mm]
56  & S^{(1,1,1)}  & \mathbf{2}_{\operatorname{std}} & S^{(2)}   \\[2mm]
120  & S^{(1,1,1)}  & \mathbf{2}_{\operatorname{std}} & S^{(1,1)} \\[2mm]
168 & S^{(2,1)}    & \mathbf{2}_{\operatorname{std}} & S^{(2)}   \\[2mm]
168 & S^{(2,1)}    & \mathbf{2}_{\operatorname{std}} & S^{(1,1)}. \\
\hline
\end{array}
\end{align}

Comparing this with our original discussion at the beginning of \autoref{sec:symmetry_reduction}, where we only decomposed with respect to the $S_2$, we have not just two large $\sim 2000\times 2000$ blocks, but rather $12$ blocks of reasonable sizes, which can be used (at least theoretically) in an SDP.


\subsubsection{Large depolarizing channel, isotropic state symmetry, and extension symmetry}
\label{subsub:large_depolarizing_werner_ext}

Let us consider the situation of a depolarizing channel 
\begin{align}
    \mathcal{N}:\mathcal{S}(\mathcal{H}_P) \to \mathcal{S}(\mathcal{H}_P)
\end{align}
of arbitrary dimension $P$. Its Choi matrix is invariant under the tensor product representation of the unitary group on $\mathcal{H}_P$ due to \autoref{prop:covariant-choi}. Thus, we are using the arguments from \autoref{prop:enlarged_invariance} and \autoref{thm:small_SDPs} to get a small SDP. With similar arguments as in \autoref{subsub:iid_werner_ext} we can show that for $n\leq 3$, i.e. at most two extensions, upper bounds can be achieved with a linear program.

\begin{proposition}[LP-bounds for depolarizing channel]\label{prop:lp_bounds_depolarizing_channel}
    Consider the global depolarizing channel on $\mathcal{N}_{\operatorname{dep}}:\mathcal{S}(\mathcal{H}_P \to \mathcal{S}(\mathcal{H}_P)$. Then, the optimization problem \autoref{thm:small_SDPs} without extensions is a linear program. With one extension, a several linear variables and four $2\times 2$ SDP variables. 
\end{proposition}

\begin{proof}
    We consider the spaces $V \coloneqq \mathcal{H}_{P \bar{P}^{(1\ldots 3)}} \otimes \mathcal{H}_{L \bar{L}^{(1\ldots 3)}}$. then we make use of the fact that we can consider $V$ as a $\mathcal{W} \coloneqq S_3\otimes \mathcal{U}(\operatorname{dim P}) \otimes \mathcal{U}(\operatorname{dim} L)$ module. In particular we can decompose 
    \begin{align}
        V \cong_{G} \mathcal{H}_{P \bar{P}^{(1\ldots 3)}} \otimes \mathcal{H}_{L \bar{L}^{(1\ldots 3)}}
    \end{align}
    with each of them as $\mathbb{C} \mathcal{W}$-modules. For the second one, i.e. $\mathcal{H}_{L \bar{L}^{(1\ldots 3)}}$ we already know from the discussion in \autoref{subsub:iid_werner_ext} how it decomposes
    \begin{align}
        \mathcal{H}_{L\bar{L}^{(1\ldots n)}} \cong_{S_{n}\times \mathcal{U}(\operatorname{dim} L )} \bigoplus_{\lambda \vdash n+1} E^\lambda \otimes \bigoplus_{A \in \operatorname{Rem}(\lambda)} S^{\lambda_A} = \bigoplus_{\lambda \vdash n+1}\bigoplus_{A \in \operatorname{Rem}(\lambda)}  E^\lambda \otimes S^{\lambda_A}.
    \end{align}
    Exactly with the same reasoning and arguments we get a decomposition for $\mathcal{H}_{P\bar{P}^{(1\ldots 3)}}$
    \begin{align}
        \mathcal{H}_{P\bar{P}^{(1\ldots n)}} \cong_{S_{n}\times \mathcal{U}(\operatorname{dim} P)} \bigoplus_{\mu \vdash n+1} E^\mu \otimes \bigoplus_{A \in \operatorname{Rem}(\mu)} S^{\mu_A} = \bigoplus_{\mu \vdash n+1}\bigoplus_{A \in \operatorname{Rem}(\mu)}  E^\mu \otimes S^{\mu_A}.
    \end{align}
    The only difference is that the dimensions of the Schur modules $E^\lambda$ respectively $E^\mu$ adopts to the dimensions $\operatorname{dim} P $ respectively $\operatorname{dim} L$. To get now a decomposition of the full $\mathbb{C} \mathcal{W}$-module, we form the tensor product
    \begin{equation}
    \begin{aligned}
        \mathcal{H}_{P\bar{P}^{(1\ldots n)}} \otimes \mathcal{H}_{L\bar{L}^{(1\ldots n)}} &\cong \bigoplus_{\mu \vdash n+1}\bigoplus_{A \in \operatorname{Rem}(\mu)}  E^\mu \otimes S^{\mu_A} \otimes \bigoplus_{\lambda \vdash n+1}\bigoplus_{A \in \operatorname{Rem}(\lambda)}  E^\lambda \otimes S^{\lambda_A} \\
        &= \bigoplus_{\mu \vdash n+1}\bigoplus_{\lambda \vdash n+1} E^\mu \otimes E^\lambda \otimes \bigoplus_{A \in \operatorname{Rem}(\mu)} \bigoplus_{B \in \operatorname{Rem}(\lambda)}   S^{\mu_A} \otimes S^{\lambda_B}
    \end{aligned}
    \end{equation}
    which yields for fixed $\mu$ and $\lambda$ the following table
\begin{align}
\begin{array}{c|ccc}
(\mu,\lambda) & (3) & (2,1) & (1,1,1) \\[6pt]
\hline
(3)
 & S^{(2)}
 & S^{(1,1)} \;\oplus\; S^{(2)}
 &  S^{(1,1)}
\\[8pt]
(2,1)
 & S^{(1,1)} \;\oplus\; S^{(2)}
 & 2 S^{(2)} \;\oplus\;2 S^{(1,1)}
 & 2 S^{(1,1)}
\\[8pt]
(1,1,1)
 & S^{(1,1)}
 &  2S^{(1,1)}
 & S^{(2)}.
\end{array}
\end{align}
From that we see that only the four combinations out of $(2,1)$ and $(1,1,1)$ get multiplicities. The statement about the LP is easily seen from the trivial case
\begin{align}
\begin{array}{c|cc}
(\mu,\lambda) & (2) & (1,1) \\[6pt]
\hline
(2)   & S^{(1)} \otimes S^{(1)} & S^{(1)} \otimes S^{(1)} \\[6pt]
(1,1) & S^{(1)} \otimes S^{(1)} & S^{(1)} \otimes S^{(1)}.
\end{array}
\end{align}

\end{proof}


\subsubsection{Identical and independent copies of a depolarizing channel and isotropic state symmetry}
\label{subsub:iid_dep_werner}

In contrast to \autoref{subsub:large_depolarizing_werner_ext} we consider in this section the $m$-fold product of single qubit depolarizing channels. Importantly, the Choi matrix of the depolarizing channel is even flip invariant. This implies in \autoref{fig:symmetry} that even flips between $(1,m+1),(2,m+2)...$ are valid symmetries.

\begin{corollary}[see \cite{Wang_2019}]
    Consider the $m$-fold depolarizing channel $\mathcal{N}_{\operatorname{dep}}^{\otimes m}:\mathcal{S}(\mathcal{H}_{P_j})\to \mathcal{S}(\mathcal{H}_{P_j})$. Then, the optimization program in \autoref{thm:small_SDPs} without extension symmetry is a linear program.
\end{corollary}

\begin{proof}
    We apply directly Schur-Weyl duality to the whole $P_{1\ldots m}\bar{P}_{1 \ldots m}$ block with respect to the $S_{2m}$, which is possible, because we the Choi matrix of the depolarizing channel is flip invariant what implies that in \autoref{fig:symmetry} we have in particular the transposition $(1,m+1),\ldots(m-1,m)(2m-1,2m)$. For example 
    \begin{align}
        (1,m+1)\;\bigl((1,2)(m+1,m+2)\bigr)\;(1,m+1) = (1,2),
    \end{align}
    which can be done for all such elements in both columns. Thus, we get $S_m \times S_m$ and with just one additional transposition $(1,m+1)$, we can apply e.g. the graph connection criterion, which implies that we can generate the whole $S_{2m}$. Thus, we have the full Schur-Weyl duality as irreducible decomposition for $S_{2m}$ with qubits $P_{1\ldots m}\bar{P}_{1 \ldots m}$. This implies directly an LP formulation, because no multiplicities occur. In particular the tensor product with the diagonal $U$ action over $L$ and $\bar{L}$ does not change anything. 
\end{proof}


\section{Numerical considerations}\label{sec:numerical_results}

We provide some explanations for numerical implementations, as things turn out to be a bit trickier than simply applying the presented symmetry reduction tools. In particular, we are concerned with the case of three iid-qubits and isotropic state symmetry as already described in \autoref{subsub:iid_werner_ext}. The reason for the difficulty is on one hand that we deal with a mixture of a compact group, $\mathcal{U}(2)$, and $S_2\times S_3$ and on the other hand the dimensions are too large for applying numerical methods such as RepLAB \cite{Rosset_2021} directly. Thus, we consider a $4096$-dimensional vector space and construct a basis for its invariant subspace under the group action of $S_2\times S_3 \times \mathcal{U}(2)$.


\subsection{Block structure and decomposition strategy}
\label{subsec:block_structure_decomposing_strategy}

Because the full $4096$-dimensional space is too large to decompose directly, the method is to first split the space as a tensor product between the $P\bar{P}^{(1,2)}$ block (a $3 \times 3$ array) and the $L\bar{L}^{(1,2)}$ systems. However, this splitting presents a \emph{basis alignment problem}: the tensor product structure does not automatically align the bases from different subspaces in a way that controls local unitary degrees of freedom. To overcome this, the strategy is to decompose the $W_i$ block
\begin{align}
\begin{array}{ccc}
    P_1 & \bar{P}_1 & \bar{P}_1^{(2)} \\
    P_2 & \bar{P}_2 & \bar{P}_2^{(2)} \\
    P_3 & \bar{P}_3 & \bar{P}_3^{(2)} 
    \end{array}
\end{align}
as a direct sum of twice itself, whereby the action of $S_2$ should be different on both blocks. One part corresponds to the subspace on which the $S_2$ acts with eigenvalue $+1$ (symmetric) and the other where $S_2$ acts with eigenvalue $-1$ (antisymmetric). Once these two blocks are obtained and further decomposed with respect to the $S_3$ symmetry (which permutes the rows), one can use a known Clebsch--Gordan basis to properly align the decomposition. This step can be done in a few hours with RepLAB \cite{Rosset_2021}.


\subsection{\texorpdfstring{Clebsch--Gordan basis and the $\overline{u}\otimes u\otimes u$ representation}{}}

To decompose the $L\bar{L}^{(1,2)}$-systems we consider a concrete realization of the Clebsch--Gordan coefficients as
\begin{align*}
v_1 &= \ket{100},\\[1mm]
v_2 &= \frac{-\ket{101} - \ket{110} + \ket{000}}{\sqrt{3}},\\[1mm]
v_3 &= \frac{-\ket{111} + \ket{001} + \ket{010}}{\sqrt{3}},\\[1mm]
v_4 &= \ket{011},\\[1mm]
v_5 &= \frac{-\ket{101} - \ket{110} - 2\ket{000}}{\sqrt{6}},\\[1mm]
v_6 &= \frac{\ket{001} + \ket{010} + 2\ket{111}}{\sqrt{6}},\\[1mm]
v_7 &= \frac{-\ket{101} + \ket{110}}{\sqrt{2}},\\[1mm]
v_8 &= \frac{\ket{010} - \ket{001}}{\sqrt{2}}.
\end{align*}
Here, the normalization factors $\sqrt{2}$, $\sqrt{3}$, and $\sqrt{6}$ ensure that each vector is of unit norm. Notice that the construction leads to a representation of the form
\begin{align}
\overline{u}\otimes u\otimes u,
\end{align}
which arises naturally from the isotropic state defined on the $L\bar{L}^{(1,2)}$-systems. Moreover, the subspace on the $L\bar{L}^{(1,2)}$ systems decomposes under $S_2$ into the symmetric subspace plus two copies of the standard two-dimensional spin representation.


\subsection{PPT and non-signaling constraints}\label{subsec:ppt_ns}

It is known that AQEC can be described as a more general channel-coding problem (see, for example, \cite{Leung_2015} and references therein). In particular, a straightforward scenario would allow arbitrary codes with unrestricted correlations between the encoder and decoder. However, in the context of approximate quantum error correction, we focus on the strongest assumptions for both encoder and decoder—namely, so-called \emph{unassisted codes}. Unassisted codes must satisfy a PPT condition as well as non-signaling assistance. To include PPT constraints, we use the following lemma (which follows by direct calculation).

\begin{lemma}\label{lem:ppt}
Let $\rho_{AB}$ be a density operator on a bipartite Hilbert space $\mathcal{H}_A \otimes \mathcal{H}_B$. Suppose there is a finite group of unitaries $\mathcal{U}(\mathcal{H})$ on $\mathcal{H}_B$ such that
\begin{align}
(U \otimes U) \,\rho_{AB}\, (U \otimes U)^\dagger = \rho_{AB}
\quad\text{for all } U \in \mathcal{U}(\mathcal{H}).
\end{align}
Then the partial transpose of $\rho_{AB}$ with respect to $B$ is invariant under $U \otimes \overline{U}$, i.e.,
\begin{align}
(U \otimes \overline{U}) \,\rho_{AB}^{T_B}\, (U \otimes \overline{U})^\dagger 
= \rho_{AB}^{T_B}
\quad\text{for all } U \in \mathcal{U}(\mathcal{H}).
\end{align}
\end{lemma}

\autoref{lem:ppt} then shows that we can implement the PPT constraints similarly to the linear Choi constraints, because it is easy to see that we can mimic \autoref{prop:enlarged_invariance} for the reduced state on the first bi-partition as well and get exactly such symmetries as described in \autoref{lem:ppt}.\footnote{This solves in particular the question whether one can impose PPT constraints in the symmetry reduced hierarchy in \cite{chee2024efficient}.} Thus, in practice we just need to use the dual representation for the basis, where a representation and its dual representation equal up to a global unitary.\\ 

The non-signaling constraints are on the other hand somehow more involved. The non-signaling constraints imply linear constraints on the objective function, so they are SDP representable. Furthermore, as observed in \cite{berta2021semidefinite}, the constraint
\begin{align}\label{eq:non_signalling_1}
    \operatorname{tr}_{L}[\rho_{LP (\bar{L}\bar{P})^{(1\ldots n)}}] = \frac{\mathds{1}_P}{d_P} \otimes \rho_{ (\bar{L}\bar{P})^{(1\ldots n)}}
\end{align}
in \autoref{thm:hierarchy_outer_approx} is exactly the first of the two non-signaling constraints (usually called non-signaling from Bob to Alice). However, the other way around, would lead to a constraint looking like
\begin{align}\label{eq:non_signalling_2}
    \operatorname{tr}_{P}[\rho_{LP (\bar{L}\bar{P})^{(1\ldots n)}}] = \rho_{LP (\bar{L}\bar{P})^{(1\ldots n-1)}}\otimes \frac{\mathds{1}_{\bar{L}^{(n)}}}{d_{\bar{L}^{(n)}}},
\end{align}
which is not required for the convergence of \autoref{thm:hierarchy_outer_approx} (called non-signaling from Alice to Bob). This difference has an impact on the numerical results. Our implementation at the moment works only with the non-signaling condition in \eqref{eq:non_signalling_1}, but with more work one could also include \eqref{eq:non_signalling_2}.


\subsection{Numerical results}\label{subsec:results}

We provide two results calculated with the techniques described in this work. We symmetrize the task of encoding one qubit into three qubits and consider the case of iid-noise. The resulting optimization problem in dimension $4096$ is then rewritten with the symmetry reduction methods described in \autoref{sec:commutative_relaxation} into $12$ blocks as described in \autoref{subsub:iid_werner_ext}. To handle the large dimensions we used the tools from \autoref{subsec:block_structure_decomposing_strategy} and to impose the PPT constraint we used the \autoref{lem:ppt} to relate it back to the usual cases of linear constraints we have already discussed in \autoref{prop:enlarged_invariance}. 

\begin{figure}
    \centering
    
\begin{tikzpicture}
\begin{axis}[
    xmin=0,
    xmax=0.7,
    width=14cm,
    height=8cm,
    xlabel={depolarizing parameter},
    ylabel={channel fidelity},
    title={Outer bounds for three qubits iid-depolarizing noise},
    legend pos=north east,
    grid=major
]

\addplot+[mark=none, color=MidnightBlue] coordinates {
(0,1.00000000000065)
(0.0344827586206897,0.99006468100788)
(0.0689655172413793,0.977592357214459)
(0.103448275862069,0.96272144900926)
(0.137931034482759,0.945590247277051)
(0.172413793103448,0.926337290762238)
(0.206896551724138,0.905100660952972)
(0.241379310344828,0.882019045478011)
(0.275862068965517,0.857230720462796)
(0.310344827586207,0.830874062077705)
(0.344827586206897,0.803087457514341)
(0.379310344827586,0.77400928697465)
(0.413793103448276,0.743777932963072)
(0.448275862068966,0.712532036426432)
(0.482758620689655,0.680409200871863)
(0.517241379310345,0.647548587479795)
(0.551724137931034,0.61408831850579)
(0.586206896551724,0.580166776005883)
(0.620689655172414,0.545922342039458)
(0.655172413793103,0.511493398686546)
(0.689655172413793,0.500000041176634)
};
\addlegendentry{Level $1$, PPT, NS (Bob $\rightarrow$ Alice)}

\addplot+[mark=none, color=PineGreen] coordinates {
(0,1.00000006981159)
(0.0344827586206897,0.988335506995324)
(0.0689655172413793,0.974380490062433)
(0.103448275862069,0.958266847580353)
(0.137931034482759,0.940115047747567)
(0.172413793103448,0.920060491428538)
(0.206896551724138,0.898229409773625)
(0.241379310344828,0.874749852253308)
(0.275862068965517,0.849750753083169)
(0.310344827586207,0.823359943925888)
(0.344827586206897,0.795702948318388)
(0.379310344827586,0.766912244285844)
(0.413793103448276,0.737113603352771)
(0.448275862068966,0.70643627114731)
(0.482758620689655,0.675007982893763)
(0.517241379310345,0.642956891760274)
(0.551724137931034,0.610412015155685)
(0.586206896551724,0.577499861488063)
(0.620689655172414,0.544349390496222)
(0.655172413793103,0.511089107598921)
(0.689655172413793,0.50000039263973)

};
\addlegendentry{Level $2$, PPT, NS (Bob $\rightarrow$ Alice)}

\end{axis}
\end{tikzpicture}
\begin{tikzpicture}
\begin{axis}[
    xmin=0,
    xmax=1,
    width=14cm,
    height=8cm,
    xlabel={damping parameter},
    ylabel={channel fidelity},
    title={Outer bounds for three qubits iid-amplitude damping noise},
    legend pos=south west,
    grid=major
]

\addplot+[mark=none, color=MidnightBlue] coordinates {
(0,1.00000000000065)
(0.0344827586206897,0.999813991131562)
(0.0689655172413793,0.999222087029469)
(0.103448275862069,0.998171484910239)
(0.137931034482759,0.996603035226908)
(0.172413793103448,0.994451484601565)
(0.206896551724138,0.991646445752033)
(0.241379310344828,0.988109681603265)
(0.275862068965517,0.983758909905463)
(0.310344827586207,0.978504465752878)
(0.344827586206897,0.972248052522798)
(0.379310344827586,0.964889893267677)
(0.413793103448276,0.95632251217589)
(0.448275862068966,0.94643099425421)
(0.482758620689655,0.93509475156418)
(0.517241379310345,0.922185118504099)
(0.551724137931034,0.907571021274835)
(0.586206896551724,0.891114802367654)
(0.620689655172414,0.872676502864069)
(0.655172413793103,0.852114751390449)
(0.689655172413793,0.829292716282071)
(0.724137931034483,0.804058816154471)
(0.758620689655172,0.776268856659787)
(0.793103448275862,0.745794266933805)
(0.827586206896552,0.712519689712587)
(0.862068965517241,0.67634539911029)
(0.896551724137931,0.637185774810756)
(0.931034482758621,0.594954870218209)
(0.96551724137931,0.549411549808639)
(1,0.500000000891859)
};
\addlegendentry{Level $1$, PPT, NS (Bob $\rightarrow$ Alice)}

\addplot+[mark=none, color=PineGreen] coordinates {
(0,1.00000006981159)
(0.0344827586206897,0.999663279903013)
(0.0689655172413793,0.998656783235379)
(0.103448275862069,0.996975629140083)
(0.137931034482759,0.994597125658287)
(0.172413793103448,0.991491715085833)
(0.206896551724138,0.987612269234304)
(0.241379310344828,0.982912659659612)
(0.275862068965517,0.977327885351905)
(0.310344827586207,0.970795439460163)
(0.344827586206897,0.963245381846756)
(0.379310344827586,0.95460129312932)
(0.413793103448276,0.944782291413351)
(0.448275862068966,0.933699607287708)
(0.482758620689655,0.921261178396304)
(0.517241379310345,0.907369718766138)
(0.551724137931034,0.89192730488573)
(0.586206896551724,0.874825654313323)
(0.620689655172414,0.855957705937527)
(0.655172413793103,0.835211040398635)
(0.689655172413793,0.812469732536723)
(0.724137931034483,0.787613432858826)
(0.758620689655172,0.760518062420908)
(0.793103448275862,0.731051827714754)
(0.827586206896552,0.699075756463849)
(0.862068965517241,0.664438353508194)
(0.896551724137931,0.626957301903186)
(0.931034482758621,0.586911650810086)
(0.96551724137931,0.544519142012839)
(1,0.500000506548027)
};
\addlegendentry{Level $2$, PPT, NS (Bob $\rightarrow$ Alice)}

\end{axis}
\end{tikzpicture}

\caption{We consider outer bounds on the channel fidelity when encoding a single qubit into three iid qubits subject to depolarizing respectively amplitude damping noise. In addition to the PPT constraint, we include a non-signaling constraint preventing signaling from Bob to Alice, as discussed in \autoref{subsec:ppt_ns}. For the first level of the hierarchy, we use the code provided in \cite{berta2021semidefinite}, but we do not include the non-signaling constraint from Alice to Bob, since the standard hierarchy \autoref{thm:hierarchy_outer_approx} does not assume it.
}
    \label{fig:dep_results}
\end{figure}
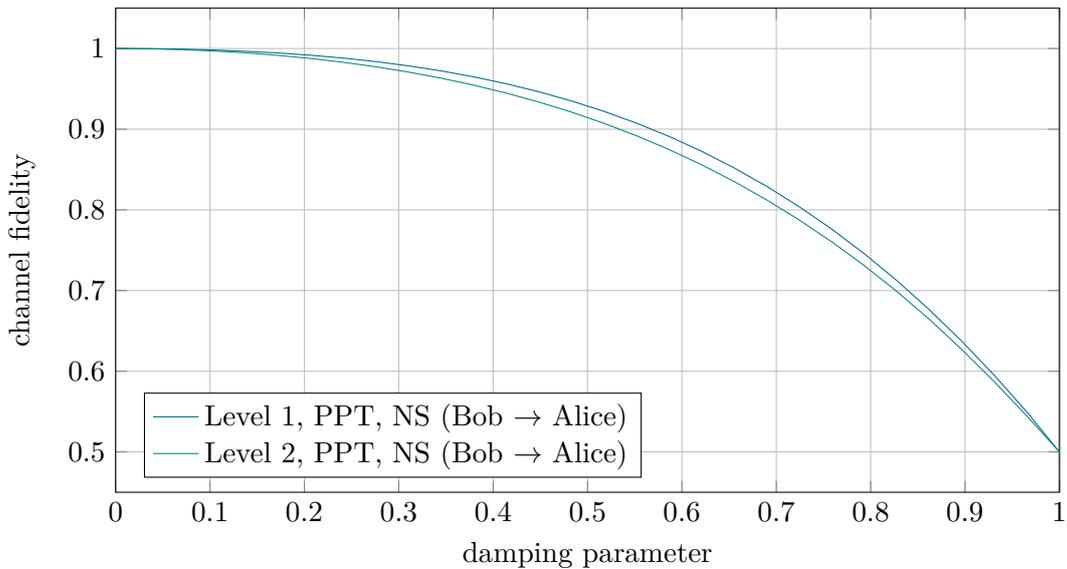


\section{Combining general symmetries of optimization variable and objective function}
\label{sec:framework}

\subsection{General framework} In this section we use the AQEC problem as and in particular its SDP relaxation as a motivation for a general discussion of the following question. 

\begin{question}\label{quest:symmetries}
Let $\mathcal{U}$ and $\mathcal{V}$ be subgroups of the unitary group $\mathcal{U}(\mathcal{H})$ and consider an SDP 
\begin{equation}\label{eq:symmtry_theorem_problem}
        \begin{aligned}
            \sup \ &\operatorname{tr}[H \rho] \\
            \operatorname{s.t.} \  &\operatorname{tr}[A_j \rho ] \leq b_j,  \quad  1\leq j \leq m \\
            &\rho \geq 0 \\
            &UHU^\dagger = H  \quad \text{and} \quad U A_j U^\dagger = A_j \quad \text{for all} \  U \in \mathcal{U}, \ 1\leq j \leq m \\
    &V\rho V^\dagger = \rho \quad \text{for all} \ V \in \mathcal{V}.
        \end{aligned}
    \end{equation}
Can we benefit from both symmetry groups $\mathcal{U}$ and $\mathcal{V}$ for the purpose of symmetry reduction of SDPs?
\end{question}

As an introductory example, we consider the action of the iid-symmetry in a slight different way as in \autoref{sec:commutative_relaxation} and the extendibility symmetry in AQEC.

\begin{example}\label{example}
    \begin{figure}[ht]
    \centering
\begin{equation}
\begin{aligned}
\begin{matrix}
P_1 & \bar{P}_1 & \bar{P}_1^{(2)}\\
P_2 & \bar{P}_2 & \bar{P}_2^{(2)}\\
L & \bar{L} & \bar{L}^{(2)} \\
\end{matrix}
\end{aligned}
\end{equation}
\caption{In this figure, corresponding to \autoref{example}, we consider a special case of \autoref{fig:symmetry} for the case of a noisy channel which is permutation invariant under an $S_2$ symmetry, given by the permutation $(P_1,P_2)(\bar{P}_1,\bar{P}_2)$ and an extension symmetry permuting the two columns $\bar{P}_1,\bar{P}_2, \bar{L}$ with the column $\bar{P}^{(2)}_1, \bar{P}_2^{(2)}, \bar{L}^{(2)}$. Together we can consider them as a subgroup of the permutation group acting on the whole figure. In \autoref{example} we investigate the group generated by both groups, the iid-symmetry and the extension symmetry.}
    \label{fig:example}
\end{figure}

This example shows that the group $\mathcal{W}$, combined by all elements of the iid-symmetry group denoted $\mathcal{U}$ (see \autoref{subsec:types_of_symmetries} (c)) and generated by the action $(P_1, P_2)(\bar{P}_1, \bar{P}_2)$, and the extendibility symmetry group denoted $\mathcal{V}$ (see \autoref{subsec:types_of_symmetries} (a)), which permutes the column $(\bar{P}_1, \bar{P}_2, \bar{L})$ with the column $(\bar{P}_1^{(2)}, \bar{P}_2^{(2)}, \bar{L}^{(2)})$, does not yield a symmetry under which the SDP \eqref{eq:optimization_channel_fidelity_symmetry} remains invariant. For simplicity, we restrict the action of the extendibility symmetry group to the systems $\bar{P}_i^{(j)}$, where $1 \leq i,j \leq 2$. Both groups, the iid-symmetry group $\mathcal{U}$ and the extendibility symmetry group $\mathcal{V}$, considered individually, are a $S_2$ and are subgroups of
\begin{align}
    \operatorname{Sym}\left(\{P_1,P_2,\bar{P}_1,\bar{P}_2,\bar{P}_1^{(2)},\bar{P}_2^{(2)}\}\right).    
\end{align}
To simplify notation we denote $\mathcal{U}$ and $\mathcal{V}$ as subgroups of $S_6$, enumerating row-wise, i.e.
\begin{align}
    P_1 \equiv 1,P_2  \equiv 4,\bar{P}_1 \equiv 2,\bar{P}_2 \equiv 5,\bar{P}_1^{(2)} \equiv 3,\bar{P}_2^{(2)} \equiv 6,
\end{align}
and restrict the action on the $P_i$ and $\bar{P}_i$-systems for simplicity. Then, we have for the iid-symmetry group $\mathcal{U}$ and the extendibility symmetry group $\mathcal{V}$ that
\begin{align}
    \mathcal{U} = \big\{\operatorname{id}, (1,4)(2,5)\big\} \quad \mathcal{V} = \big\{\operatorname{id}, (2,3)(5,6)\big\}.
\end{align}
Multiplying all elements of $\mathcal{U}$ with all elements from $\mathcal{V}$ yields a large group $\mathcal{W}$ as
\begin{equation}
    \begin{aligned}
        \mathcal{W} =  \big\{&\operatorname{id},(2,3)(5,6),(2,5)(3,6),(2,6)(3,5),(1,4)(3,6),(1,4)(2,3,5,6),\\
    &(1,4)(2,5),(1,4)(2,6,5,3)\big\}.
    \end{aligned}
\end{equation}
Returning to \eqref{eq:optimization_channel_fidelity_symmetry}, one could consider the groups $\mathcal{U}$ and $\mathcal{V}$ individually, allowing the application of the standard framework from \autoref{subsec:recap_standard_symmetry} to each group separately. However, to incorporate both symmetry groups $\mathcal{U}$ and $\mathcal{V}$ simultaneously, we must address in what sense might the resulting SDP be invariant under the combined group $\mathcal{W}$.

Towards this goal, let us assume we have an objective $H$ and a variable $\rho$ satisfying the following symmetry conditions
\begin{align}
    \operatorname{tr}[UHU^\dagger \rho] = \operatorname{tr}[H\rho] \quad \text{and} \quad \operatorname{tr}[H V\rho V^\dagger] = \operatorname{tr}[H\rho]
\end{align}
for the unitary groups $U$ and $V$. Combining these two equations yields that we aim for solutions invariant under 
\begin{align}\label{eq:joint_symmetry_condition}
    \operatorname{tr}[H U^\dagger V \rho (U^\dagger V)^\dagger] = \operatorname{tr}[H \rho] \quad U \in \mathcal{U},V \in \mathcal{V}.
\end{align}
This means that exactly all products of this form are allowed. If now all elements $\zeta \in \mathcal{W}$ would have a representation in the form 
\begin{align}
    \zeta = \sigma \tau \quad \sigma \in \mathcal{U}, \ \tau \in \mathcal{V}
\end{align}
we would have a joint symmetry as shown in \eqref{eq:joint_symmetry_condition}. However, if we consider the element $(1,4)(2,3,5,6)$, then it has a decomposition $e \coloneqq  (2,3)(5,6)(1,4)(2,5)$, which is equivalent to a representation $\tau \in \mathcal{V}$ and $\sigma \in \mathcal{U}$ such that 
\begin{align}
   e =  \tau \sigma .
\end{align}
Assuming that we can write 
\begin{align}
    e = \tilde{\sigma} \tilde{\tau} \ \text{for some} \ \tilde{\sigma} \in \mathcal{U}, \ \tilde{\tau} \in \mathcal{V},
\end{align}
we could write
\begin{align}
    \text{$\tau \sigma = \tilde{\sigma} \tilde{\tau} \ \Rightarrow \ \tau = \tilde{\sigma} \tilde{\tau} \sigma \in \mathcal{V}$ (because $\sigma$ is of order $2$).}
\end{align}
However, we have only have two possibilities for $\tilde{\tau}$, either $(2,3)(5,6)$ or the identity. The identity cannot be the correct choice because we want to show a representation for $\tau$ and in this case the underlying element would be in $\mathcal{U}$. Similarly we only have the possibilities $(1,4)(2,5)$ or the identity $\tilde{\sigma}$. Thus, we find 
\begin{equation}
\begin{aligned}
    \underbrace{(1,4)(2,5)}_{= \tilde{\sigma}}\underbrace{(2,3)(5,6)}_{= \tilde{\tau}}\underbrace{(1,4)(2,5)}_{=\sigma} = (2,6)(3,5) \notin \mathcal{V} \\
    \underbrace{\operatorname{id}}_{= \tilde{\sigma}}\underbrace{(2,3)(5,6)}_{= \tilde{\tau}}\underbrace{(1,4)(2,5)}_{=\sigma} = (1,4)(2,3,5,6) \notin \mathcal{V}
\end{aligned}
\end{equation}
and we do not get the desired element in $\mathcal{V}$. This example shows that the problem is not invariant under the symmetry group $\mathcal{W}$. 
\end{example}

Importantly, the key difference between \autoref{example} and \autoref{sec:commutative_relaxation} lies in the \emph{manner in which} the iid-symmetry acts. From the abstract perspective of group theory, this corresponds to considering two different representations of $S_m$ on a Hilbert space. As we observe in \autoref{example}, the interaction of the images of $S_m$ as a subgroup of $\mathcal{U}(\mathcal{H})$ leads to different results, and consequently to a different interplay with the extendibility symmetry: in the commutative action of \autoref{sec:commutative_relaxation}, the two groups commute, whereas in \autoref{example}, the two group actions do not commute. This observation is somehow generic and can be generalized to an extension of the symmetry reduction framework of \cite{Klerk2007ReductionOS} towards a notion of \emph{joint symmetry reduction} giving an answer on \autoref{quest:symmetries}, what is purpose of this section. 

Motivated by \autoref{example} and \eqref{eq:joint_symmetry_condition}, one could try to define joint symmetry reduction, with the following notion of invariance
\begin{equation}
    \begin{aligned}
    \operatorname{tr}[H \rho] &= \operatorname{tr}[H U \rho U^\dagger] \quad U\in \mathcal{U}, \ \rho \in \mathcal{M}^\prime_{\mathcal{V}}(\mathcal{H}) \\
        \operatorname{tr}[A_j \rho] &= \operatorname{tr}[A_j U \rho U^\dagger] \quad U \in \mathcal{U}, \ \rho \in \mathcal{M}^\prime_{\mathcal{V}}(\mathcal{H}), \ 1 \leq j \leq m.
    \end{aligned}
\end{equation}
The problem with this notion of invariance is that, in general, $\mathcal{U}$ and $\mathcal{V}$ can generate elements that do not preserve the \emph{inner products} (see \autoref{example}). In fact, given symmetries $\mathcal{U}$ and $\mathcal{V}$, we are looking for a method to construct a compact group $\mathcal{W}$ from $\mathcal{U}$ and $\mathcal{V}$ such that
\begin{equation}\label{eq:joint_symmetries_weak_form}
    \begin{aligned}
        \operatorname{tr}[\rho H] & = \operatorname{tr}[H W \rho W^\dagger] \quad \text{for all } W \in \mathcal{W}, \\
        \operatorname{tr}[\rho A_j] & = \operatorname{tr}[A_j W \rho W^\dagger] \quad \text{for all } W \in \mathcal{W}, \ 1 \leq j \leq m.
    \end{aligned}
\end{equation}
In this case, we can replace any optimizer $\rho^\star$ of the original SDP \eqref{eq:symmtry_theorem_problem}
\begin{equation}\label{eq:optimizer_without_sym}
    \begin{aligned}
        \rho^\star \coloneqq \arg \max \ & \operatorname{tr}[H \rho] \\
        \text{s.t.} \quad & \operatorname{tr}[A_j \rho ] \leq b_j, \quad 1 \leq j \leq m, \\
        & \rho \geq 0
    \end{aligned}
\end{equation}
with its twirled operator
\begin{align}
    \rho^\star \mapsto \int_\mathcal{W} d\mu(W) \, W \rho^\star W^\dagger,
\end{align}
where $\mu$ is the unique normalized Haar measure corresponding to $\mathcal{W}$. The optimization task can then be solved within the commutant $\mathcal{M}^\prime_{\mathcal{W}}(\mathcal{H})$. Observe that this is, in general, a much weaker\,---\,but still sufficient\,---\,assumption compared to the standard symmetry reduction framework discussed in \autoref{subsec:recap_standard_symmetry}. Specifically, we neither require conditions such as $W H W^\dagger = H$ and $W A_j W^\dagger = A_j$ for all $W \in \mathcal{W}$ and $1 \leq j \leq m$, nor do we initially enforce $W \rho W^\dagger = \rho$ for all $\rho \in \mathcal{M}^\prime_{\mathcal{V}}(\mathcal{H})$. Concluding this line of thought, we arrive at what is arguably the most general perspective on this mathematical problem: namely, the setting of two unital $C^\star$-algebras on $\mathcal{H}$, denoted by $\mathcal{A}$ and $\mathcal{B}$, along with the (compact) groups of unitaries in their commutants, $\mathcal{U}_{\mathcal{A}^\prime}$ and $\mathcal{U}_{\mathcal{B}^\prime}$. The goal is to construct the largest group $\mathcal{W}$ out of $\mathcal{U}_{\mathcal{A}^\prime}$ and $\mathcal{U}_{\mathcal{B}^\prime}$ such that, whenever $H, A_1, \ldots, A_m \in \mathcal{A}$ and $\rho \in \mathcal{B}$, the conditions in \eqref{eq:joint_symmetries_weak_form} are satisfied. 

In general, the question seems challenging, which leads us to propose a simplified notion of joint symmetry as follows.

\begin{definition}[Joint symmetry reduction]\label{def:strong_joint_symmetry}
        Consider the SDP in \autoref{quest:symmetries}. We say that the SDP exhibits a \emph{joint symmetry} if the set 
        \begin{align}\label{eq:strong_joint_symmetry_def}
            \mathcal{W} \coloneqq \{UV \ \vert \ U \in \mathcal{U}, \ V \in \mathcal{V}\}
        \end{align}
        is a group. 
\end{definition}

It is evident that $\mathcal{W}$ is a set, and thus \eqref{eq:strong_joint_symmetry_def} is well-defined. The precise sense in which the definition of $\mathcal{W}$ as a group is well-behaved is clarified in the following lemma.

\begin{lemma}\label{lem:strong_symmetry}
    Consider the set $\mathcal{W}$ as defined in \eqref{eq:strong_joint_symmetry_def}. Then the following properties hold:
    \begin{enumerate}
        \item[(a)] If $\mathcal{U}$ and $\mathcal{V}$ are closed subgroups, then the set $\mathcal{W}$ is compact with respect to the operator norm topology.
        
        \item[(b)] If $\mathcal{U}$ and $\mathcal{V}$ commute and have trivial intersection, i.e., $\mathcal{U} \cap \mathcal{V} = \{\operatorname{id}\}$, then $\mathcal{W} \cong \mathcal{U} \times \mathcal{V}$, the direct product of $\mathcal{U}$ and $\mathcal{V}$.
        
        \item[(c)] The inner products in \eqref{eq:joint_symmetries_weak_form} are invariant under all elements of $\mathcal{W}$.
    \end{enumerate}
\end{lemma}

\begin{proof}
    \begin{enumerate}
        \item[(a)] Consider the map
        \begin{align}\label{eq:proof_lemma_strong_symmetry}
            \phi: \mathcal{U} \times \mathcal{V} \to \mathcal{U}(\mathcal{H}), \quad (U,V) \mapsto UV.
        \end{align}
        As closed subsets of a compact set, the unitary group $\mathcal{U}(\mathcal{H})$, $\mathcal{U}$ and $\mathcal{V}$ are compact. 
        By Tychonoff's theorem, the product $\mathcal{U} \times \mathcal{V}$ is compact in the product topology. Since the unitary group is a topological group, multiplication is continuous. Furthermore, the image of a compact set under a continuous map is compact, so $\mathcal{W}$ is compact under the stated assumptions.
        
        \item[(b)] Under the assumption that $\mathcal{U}$ and $\mathcal{V}$ commute, the mapping \eqref{eq:proof_lemma_strong_symmetry} is a group homomorphism, as
        \begin{equation}
        \begin{aligned}
            \phi((U_1,V_1)\star(U_2,V_2)) &= \phi(U_1 U_2, V_1 V_2) \\
            &= U_1 U_2 V_1 V_2 \\
            &= U_1 V_1 U_2 V_2 \\
            &= \phi(U_1, V_1)\phi(U_2, V_2).
        \end{aligned}
        \end{equation}
        The image of $\phi$ is clearly $\mathcal{W}$, and since $\mathcal{U} \cap \mathcal{V} = \{\operatorname{id}\}$, the kernel is trivial. Hence, $\phi$ is a bijective group homomorphism.

        To show that $\phi$ is an isomorphism of topological groups, we must verify that it is a homeomorphism. Since left and right translations are homeomorphisms in a topological group \cite[chap. 1.2]{Osborne2014}, and the multiplication map is open (as open sets map to open sets), it follows that $\phi$ is an open mapping. Combining this with continuity and bijectivity, we conclude that $\phi$ is a homeomorphism, and hence an isomorphism of topological groups, proving the claim.

        \item[(c)] Let $W \in \mathcal{W}$ with decomposition $W = UV$ by definition. Then, we have
        \begin{equation}
        \begin{aligned}
            \operatorname{tr}[H W\rho W^\dagger] &= \operatorname{tr}[H UV \rho (UV)^\dagger] \\
            &= \operatorname{tr}[H UV \rho V^\dagger U^\dagger] \\
            &= \operatorname{tr}[U^\dagger H U \, V\rho V^\dagger] \\
            &= \operatorname{tr}[H \rho],
        \end{aligned}
        \end{equation}
        where the last step uses that $H \in \mathcal{M}^\prime_{\mathcal{U}}(\mathcal{H})$ and $\rho \in \mathcal{M}^\prime_{\mathcal{V}}(\mathcal{H})$. An analogous argument applies to each $A_j$ for $1 \leq j \leq m$.
    \end{enumerate}
\end{proof}

A more group-theoretical perspective on $\mathcal{W}$ is provided by the natural notion of complements. As shown in \autoref{lem:strong_symmetry}~(a), given that $\mathcal{W}$ is a group with $\mathcal{U}$ as a closed subgroup, we can equivalently interpret the construction as $\mathcal{U}$ admitting a complement $\mathcal{V}$ (assuming $\mathcal{U} \cap \mathcal{V} = \{\operatorname{id}\}$). In other words, we can reformulate the question in group-theoretical terms as asking whether $\mathcal{U}$ and $\mathcal{V}$ are complements of each other. In general, this question can be addressed using the concept of the \emph{Zappa--Szép product} \cite{Brin2005}, which is arguably difficult to classify in full generality. For our purposes, we focus on interesting special cases of the Zappa--Szép product in the following proposition and relate them to our notion of symmetry reduction introduced in \autoref{def:strong_joint_symmetry}.

\begin{proposition}[Symmetry reduction for combined symmetries]\label{thm:symmetry_reduction_combined}
    Let $\mathcal{U}$ and $\mathcal{V}$ be two closed subgroups of the unitary group $\mathcal{U}(\mathcal{H})$. Then, the following are true: 
   \begin{enumerate}
        \item[(a)] $\mathcal{W}$ as defined in \autoref{def:strong_joint_symmetry} is a compact subgroup of $\mathcal{U}(\mathcal{H})$ if and only if 
        \begin{align} \label{eq:equality_sets}
            \mathcal{U}\mathcal{V}  = \mathcal{V}\mathcal{U} 
        \end{align}
        as sets (see e.g. \cite{Robinson1996}).\footnote{Recall the definition $
            \mathcal{U}\mathcal{V} \coloneqq \{ UV \mid U \in \mathcal{U},\ V \in \mathcal{V} \}$ and similarly for $\mathcal{V}\mathcal{U} $.}
        \item[(b)] Assume that $\mathcal{V}$ normalizes $\mathcal{U}$ and $\mathcal{U} \cap \mathcal{V} = \{\operatorname{id}\}$ (a similar result is of course true if $\mathcal{U}$ normalizes $\mathcal{V}$). Then we have that $\mathcal{W}$ is a group and
        \begin{align}
            \mathcal{W} \cong \mathcal{U} \rtimes \mathcal{V},
        \end{align}
        a semidirect product (see e.g. \cite{Robinson1996}). 
        \item[(c)] 
       Let $\mathcal{U}$ and $\mathcal{V}$ be finite groups. 
        If $\mathcal{W} = \mathcal{U}\mathcal{V}$ forms a group, then the following product formula holds:
        \begin{align}
            |\mathcal{W}| \cdot |\mathcal{U} \cap \mathcal{V}| = |\mathcal{U}| \cdot |\mathcal{V}|.
        \end{align}
        \item[(d)] If $\mathcal{W}$ is a group, then the SDP
\begin{equation}\label{eq:symmetry_reduced_programm_prop}
    \begin{aligned}
        \sup \ &\operatorname{tr}[H \rho] \\
        \text{s.t.} \quad &\operatorname{tr}[A_j \rho ] \leq b_j, \quad 1 \leq j \leq m, \\
        &\rho \geq 0, \\
        &\rho \in \mathcal{M}^\prime_{\mathcal{W}}(\mathcal{H})
    \end{aligned}
\end{equation}
admits the same optimal value as the SDP in \eqref{eq:symmtry_theorem_problem}, provided that the groups $\mathcal{U}$ and $\mathcal{V}$ are interpreted as in \autoref{quest:symmetries} and is in particular a joint symmetry in the sense of \autoref{def:strong_joint_symmetry}.
    \end{enumerate}
\end{proposition}

\begin{proof}
\begin{enumerate}
    \item[(a)] Let $\mathcal{W}$ be a subgroup of the unitary group $\mathcal{U}(\mathcal{H})$. Then, in particular, every element $(UV) \in \mathcal{W}$ has an inverse given by $(UV)^{-1} = V^{-1} U^{-1}$. From this, we immediately deduce that $\mathcal{U}\mathcal{V} = \mathcal{V}\mathcal{U}$.

Now assume that this equality of sets holds. To show that $\mathcal{W}$ is a group, it remains to verify that $\mathcal{W}$ is closed under multiplication, that every element has an inverse, and that $\mathcal{W}$ is topologically closed. All other group properties are inherited from the fact that the unitary group $\mathcal{U}(\mathcal{H})$ satisfies the group axioms.

For the inverses, take any $UV \in \mathcal{W}$. Then $(UV)^{-1} = V^{-1} U^{-1} \in \mathcal{V} \mathcal{U}$, and by the assumed equality $\mathcal{U} \mathcal{V} = \mathcal{V} \mathcal{U}$, we conclude that $(UV)^{-1} \in \mathcal{U} \mathcal{V}$, so $(UV)^{-1} \in \mathcal{W}$.

Next, consider two elements $U_1 V_1, U_2 V_2 \in \mathcal{W}$. Then, we get
\begin{align}
    U_1 V_1 U_2 V_2 = U_1 \underbrace{(V_1 U_2)}_{=W \in \mathcal{V} \mathcal{U}} V_2.
\end{align}
By assumption, $W \in \mathcal{V} \mathcal{U} = \mathcal{U} \mathcal{V}$, so we can write $W = \tilde{U}_2 \tilde{V}_1$ with $\tilde{U}_2 \in \mathcal{U}$ and $\tilde{V}_1 \in \mathcal{V}$. Hence, we have
\begin{equation}
\begin{aligned}
    U_1 V_1 U_2 V_2 = U_1 \tilde{U}_2 \tilde{V}_1 V_2
    = \underbrace{(U_1 \tilde{U}_2)}_{\in \mathcal{U}} \underbrace{(\tilde{V}_1 V_2)}_{\in \mathcal{V}} \in \mathcal{W},
\end{aligned}
\end{equation}
where we have used that both $\mathcal{U}$ and $\mathcal{V}$ are groups and thus closed under multiplication.

It follows that $\mathcal{W}$ is a group. Finally, to conclude that $\mathcal{W}$ is topologically closed—and hence compact\,---\,we invoke \autoref{lem:strong_symmetry}~(a).
\item[(b)] In order to prove that $\mathcal{W}$ is a compact group, we need to verify that it is closed under multiplication and inverses, and that it is topologically closed. 

Take any two elements $U_1 V_1, U_2 V_2 \in \mathcal{W}$ and compute
\begin{equation}
\begin{aligned}
    (U_1 V_1)(U_2 V_2) = U_1 \underbrace{(V_1 U_2 V_1^{-1})}_{=\tilde{U}_2 \in \mathcal{U}} V_1 V_2= (U_1 \tilde{U}_2)(V_1 V_2) \in \mathcal{W},
\end{aligned}
\end{equation}
where in the first step we used the assumption that $\mathcal{V}$ normalizes $\mathcal{U}$, i.e., 
\begin{align}
    V_1 U_2 V_1^{-1} \in \mathcal{U},    
\end{align}
and in the second step that both $\mathcal{U}$ and $\mathcal{V}$ are groups and thus closed under multiplication.

To show that $\mathcal{W}$ is closed under inverses, take any $UV \in \mathcal{W}$. Then, we get
\begin{equation}
\begin{aligned}
    (UV)^{-1} = V^{-1} U^{-1} = \underbrace{(V^{-1} U^{-1} V)}_{= \tilde{U} \in \mathcal{U}} V^{-1} = \tilde{U} V^{-1} \in \mathcal{W},
\end{aligned}
\end{equation}
where, again, we use that $\mathcal{V}$ normalizes $\mathcal{U}$ and that $\mathcal{V}$ is a group.

To conclude that $\mathcal{W}$ is compact, we refer to \autoref{lem:strong_symmetry}~(a), which guarantees compactness under the given assumptions.

To see that the resulting group is a semidirect product, we construct a concrete group homomorphism and prove its bijectivity. Define a binary operation on the set $\mathcal{U} \times \mathcal{V}$ by:
\begin{align}
    (U_1, V_1) \star (U_2, V_2) \coloneqq (U_1 V_1 U_2 V_1^{-1}, V_1 V_2).
\end{align}
Since $\mathcal{V}$ normalizes $\mathcal{U}$, this is a well-defined group operation. It is a standard result that $\mathcal{U} \times \mathcal{V}$ becomes a group under this operation, known as the semidirect product, because the map
\begin{align}
    \mathcal{V} \to \operatorname{Aut}(\mathcal{U}), \quad V \mapsto \big(U \mapsto V U V^{-1} \big)
\end{align}
is a well-defined group homomorphism. This defines the semidirect product $\mathcal{U} \rtimes \mathcal{V}$.

Now, consider the map
\begin{align}
    \phi : \mathcal{U} \rtimes \mathcal{V} \to \mathcal{W}, \quad (U, V) \mapsto UV.
\end{align}
This map is well-defined since $\mathcal{W}$ consists precisely of elements of the form $UV$ with $U \in \mathcal{U}$ and $V \in \mathcal{V}$, and we have already shown that $\mathcal{W}$ is a group.

To prove that $\phi$ is a homomorphism, we compute:
\begin{equation}
\begin{aligned}
    \phi(U_1, V_1) \cdot \phi(U_2, V_2) &= U_1 V_1 U_2 V_2 \\
    &= U_1 (V_1 U_2 V_1^{-1}) V_1 V_2 \\
    &= \phi(U_1 V_1 U_2 V_1^{-1}, V_1 V_2) \\
    &= \phi((U_1, V_1) \star (U_2, V_2)).
\end{aligned}
\end{equation}
Thus, $\phi$ is a group homomorphism. By construction, it is surjective. Moreover, since $\mathcal{U} \cap \mathcal{V} = \{\operatorname{id}\}$, it is injective. Hence, $\phi$ is a group isomorphism.

To conclude that $\phi$ is also a topological isomorphism, we recall from the proof of \autoref{lem:strong_symmetry}~(b) that the multiplication map is open. Therefore, $\phi$ is a homeomorphism, and we conclude that
\begin{align}
    \mathcal{W} \cong \mathcal{U} \rtimes \mathcal{V}.
\end{align}

\item[(c)] 
As already established in the proof of part~(b), the map
\begin{align}
    \phi: \mathcal{U} \times \mathcal{V} \to \mathcal{W}
\end{align}
is surjective. Now, fix $(U_1, V_1) \in \mathcal{U} \times \mathcal{V}$ such that $\phi(U_1, V_1) = W$ for some fixed $W \in \mathcal{W}$. We define the map
\begin{align}
    \psi_W: \mathcal{U} \cap \mathcal{V} \to \phi^{-1}(W), \quad X \mapsto (U_1 X, X^{-1} V_1).
\end{align}

To prove that $\psi_W$ is surjective, consider any pair $(U_2, V_2) \in \phi^{-1}(W)$, i.e., such that $\phi(U_2, V_2) = W = U_1 V_1$. Then, we have
\begin{align}
    U_2 V_2 = U_1 V_1 
    \quad \Rightarrow \quad 
    U_2^{-1} U_1 = V_2 V_1^{-1},
\end{align}
where both sides belong to $\mathcal{U}$ and $\mathcal{V}$, respectively. Hence, their common element lies in $\mathcal{U} \cap \mathcal{V}$.

Define
\begin{align}
    X := (U_2^{-1} U_1)^{-1} = U_1^{-1} U_2 
    \quad \text{and} \quad 
    X^{-1} = V_2 V_1^{-1}
\end{align}
and compute
\begin{equation}
\begin{aligned}
    \psi_W(X) = (U_1 X, X^{-1} V_1) = (U_1 U_1^{-1} U_2, V_2 V_1^{-1} V_1) = (U_2, V_2),
\end{aligned}
\end{equation}
which shows that $\psi_W$ is surjective.

To prove injectivity, assume $\psi_W(X) = \psi_W(Y)$ for $X, Y \in \mathcal{U} \cap \mathcal{V}$. Then, we get
\begin{equation}
\begin{aligned}
    (U_1 X, X^{-1} V_1)= (U_1 Y, Y^{-1} V_1)\quad
    \Rightarrow \quad U_1 X = U_1 Y \quad \Rightarrow \quad X = Y,
\end{aligned}
\end{equation}
since left multiplication in a group is injective. Thus, $\psi_W$ is a bijection.

Because the pre-images of distinct elements under $\phi$ are disjoint, we may count the total number of elements in the two ways
\begin{align}
    |\mathcal{U} \times \mathcal{V}| = |\mathcal{U}| \cdot |\mathcal{V}| = |\mathcal{W} | \cdot |\mathcal{U} \cap \mathcal{V}|,
\end{align}
which rearranges to the desired formula:
\begin{align}
    |\mathcal{W}| = \frac{|\mathcal{U}| \cdot |\mathcal{V}|}{|\mathcal{U} \cap \mathcal{V}|}.
\end{align}

\item[(d)]
First of all, if $\mathcal{W}$ is a group, it is automatically compact by \autoref{lem:strong_symmetry}~(a), and thus admits a unique normalized Haar measure \cite{Cohn2013}. Applying \autoref{lem:strong_symmetry}~(c), we find that the inner products remain invariant under all elements of $\mathcal{W}$. Therefore, using the definition of $\rho^\star$ from \eqref{eq:optimizer_without_sym}, we obtain
\begin{equation}
\begin{aligned}
    \operatorname{tr}[H \rho^\star] 
    = \int_{\mathcal{W}} d\mu(W) \operatorname{tr}[H W \rho^\star W^\dagger] = \operatorname{tr}\left[H \int_{\mathcal{W}} d\mu(W) \, W \rho^\star W^\dagger \right],
\end{aligned}
\end{equation}
and analogously for each $A_j$, $1 \leq j \leq m$. Since the twirled state
\begin{align}
    \int_{\mathcal{W}} d\mu(W) \, W \rho^\star W^\dagger \in \mathcal{M}_{\mathcal{W}}^{\prime}(\mathcal{H}) \subset \mathcal{B}(\mathcal{H}),
\end{align}
we conclude that the optimization problem in \eqref{eq:symmetry_reduced_programm_prop} attains the same optimal value as the original problem in \eqref{eq:symmtry_theorem_problem}.
\end{enumerate}
\end{proof}


\subsection{On extendibility- and iid-symmetry in the standard framework}
\label{subsec:adding_and_internal_symmetry}

In this section we aim to generalize the argument from \autoref{example} to the general case of $n$-extensions with extendibility symmetry $S_n$, together with an $S_m$ iid-symmetry in the noisy channel. For simplicity we neglect the systems $L\bar{L}^{(1\ldots n)}$ in \autoref{fig:symmetry} as we have already done in \autoref{example}. As a first step, we provide a structural description of the potential symmetry group $\mathcal{W}$, consisting of all combinations of elements from the iid-symmetry group if it would act as in \autoref{example} and the extendibility symmetry group. 

\begin{theorem}\label{thm:global_symmetry_group}
    The permutation group $G$ acting on \autoref{fig:symmetry}, while assuming an iid-symmetry of the
    Choi matrix $C_{\mathcal{M}^{\otimes m}}$ of a tensor power of a channel $\mathcal{M}$, 
    together with an extension symmetry $S_n$ 
    is a group extension
    \begin{align}\label{eq:group_extension_global_group}
        A_m \to G \to S_m \wr S_n,
    \end{align}
    where $A_m$ denotes the alternating group of order $m!/2$. In particular, the group extension \eqref{eq:group_extension_global_group} splits, so that we have
    \begin{align}
        G \cong A_m \rtimes (S_m \wr S_n).
    \end{align}
\end{theorem}

\begin{proof}
By assumption, every element $g \in G$ is a permutation on the subsystems 
\begin{align}
\{\bar{P}_{1\ldots m}, \dots,  \bar{P}^{(n)}_{1\ldots m}\},
\end{align}
as well as on the main system $P_{1\ldots m}$ consisting of $m$ copies. 
We define a group homomorphism
\begin{align}
   \varphi \colon G \longrightarrow \mathrm{Sym}\bigl(\bar{P}_{1\ldots m} \times \cdots \times \bar{P}^{(n)}_{1\ldots m}\bigr)
\end{align}
by letting $\varphi(g)$ be the \emph{restriction} of the permutation $g$ to the $\bar{P}_{1 \ldots m}^{(i)}$-systems.
This induced action is well-defined and surjects onto the wreath product $S_m \wr S_n$. 

Indeed, the group $S_m \wr S_n$ can be seen as ``all ways'' to permute $m$ slots within each of the $n$ blocks 
and then permute the $n$ blocks themselves. The global group $G$ must capture permutations of the $P_{1\ldots m}^{(i)}$-system ($m$ copies) combined consistently with permutations of the $n$ extensions. Thus, we get $\operatorname{Im}(\varphi)\cong S_m \wr S_n$. By definition, we have
  \begin{align}
    \ker(\varphi) = \{g \in G \ \vert \ \varphi(g) = \mathrm{id}\}
    =
    \bigl\{ g \in G \ \vert \  \text{$g$ acts trivially on all $\bar{P}_{1\ldots m}^{(i)}$} \bigr\}.
  \end{align}
Hence, every $g \in \ker(\varphi)$ must act only on the main system $P_{1\ldots m}$. Now, let $g|_{A} \in S_m$ be the permutation induced by $g$ on the $m$ slots. It is straightforward to produce elements of the form (just consider usual elements in the $S_m$ action $(x,x,1,\ldots,1)$ and act with the $S_n$ on it)
\begin{align}
  (x,x,1,\ldots,1)
  \quad\text{and}\quad
  (y,1,y,1,\ldots,1)
\end{align}
in $G$, where $x,y \in S_m$.  Their commutator,
\begin{align}
  (x,x,1,\ldots,1)\bigl(y,1,y,\ldots,1\bigr)(x,x,\ldots)^{-1}\bigl(y,1,y,\ldots,1\bigr)^{-1},
\end{align}
simplifies to
\begin{align}
  \bigl([x,y],1,\ldots,1\bigr),
\end{align}
which clearly lies in the kernel, as it acts trivially on $B_1^n$.  Since the commutators of $S_m$ generate the alternating group $A_m$, we deduce
\begin{align}
  A_m \subseteq \ker\bigl(G \to S_m \wr S_{n-1}\bigr).
\end{align}

Conversely, suppose $g=(x,1,\ldots,1)$ lies in the kernel, meaning that $g$ acts trivially on the product $B_1^n$.  Consider the standard sign homomorphism
\begin{align}
  \mathrm{sgn}\colon S_m \longrightarrow \{\pm1\}
\end{align}
restricted to $A$.  All permutations of $G$ that interchange only the remaining copies lie in the kernel of $\mathrm{sgn}$ on the first copy, so we can write 
\begin{align}
  \mathrm{sgn}(g) = \mathrm{sgn}(h)
\end{align}
for some $h \in S_m$.  However, because $g$ has sign $+1$ on the first copy of $\bar{P}_{1\ldots m}$ (it acts trivially there), it follows that $\mathrm{sgn}(g)=1$.  Hence $h$ must be an even permutation, implying $x \in A_m$.  Thus, $g \in A_m$.  Therefore,
\begin{align}
  \ker\bigl(G \to S_m \wr S_{n-1}\bigr) \subseteq A_m.
\end{align}

Combining both inclusions, we conclude
\begin{align}
  \ker\bigl(G \to S_m \wr S_{n-1}\bigr) = A_m,
\end{align}
as claimed.

Putting the above together, we have a short exact sequence:
\begin{align}
   1 \longrightarrow A_m \longrightarrow G \xrightarrow{\varphi} S_m \wr S_n \longrightarrow 1.
\end{align}
The map $\varphi$ is surjective (its image is the wreath product), 
and the kernel is exactly $A_m$. Thus $G$ is a group extension of $S_m \wr S_n$ by $A_m$.

We conclude that the global symmetry group $G$ factors as 
\begin{align}
   A_m \longrightarrow G \longrightarrow S_m \wr S_n,
\end{align}
establishing the desired extension. 

For the argument that the extension splits, we recall that we can easily generate the subgroup 
\begin{align}
    \{1\}\times A_m^{(1)} \times \cdots \times A_m^{(n)} \leq G    
\end{align}
exactly similar to the commutator argument given at the beginning of the proof. Observe that those are $n$ independent groups of the isomorphism type $A_m$. Adding the $S_n$ action, we find 
\begin{align}
    (\{1\}\times A_m^{(1)} \times \cdots \times A_m^{(n)}) \rtimes S_n \leq G,
\end{align}
whereby the action of $S_n$ can be seen as permuting the $A_m$'s. 

Now, adding on each block a transposition $\sigma_i$ which is the parallel permutation in the $P_{1\ldots m}$ block and the $\bar{P}_{1\ldots m}^{(i)}$ block, generates a set of $n$ groups isomorphic to $S_m$ groups as a direct product (because the alternating group and one transposition generates the symmetric group. A simple argument would be that the alternating group has index $2$ and the group generated by $\langle A_m,\sigma\rangle$ is strictly larger than the alternating group, because a transposition $\sigma \notin A_m$. Thus it has to be the full symmetric group). On those $S_m$ the usual $S_n$ acts as interchanging the columns in \autoref{fig:symmetry}. Thus the corresponding subgroup is isomorphic to a wreath product $S_m \wr S_n$ and we can construct a homomorphism 
\begin{align}
    \gamma: S_m \wr S_n \to G,
\end{align}
which maps onto the  $S_n$ and $A_m$ in each factor and just adds on the generating transposition in $S_m$ the corresponding parallel execution on the $P_{1\ldots m}$ block. Then, $\beta\circ \gamma = \operatorname{id}_{S_m \wr S_n}$ such that the extension splits. 
\end{proof}

From this \autoref{thm:global_symmetry_group} we conclude the following corollary, which does not fit \autoref{def:strong_joint_symmetry}.

\begin{corollary}\label{cor:no_go}
    The representation of the group $G \cong A_m \rtimes (S_m \wr S_n)$ on the Hilbert space $\mathcal{H}_{P_{1\ldots m}} \bigotimes_{i=1}^n \mathcal{H}_{P_{1\ldots m}^{(i)}}$ is not suitable for joint symmetry reduction in the sense of \autoref{def:strong_joint_symmetry}.
\end{corollary}
\begin{proof}
    The \autoref{thm:symmetry_reduction_combined}~(c) states that if a finite group $\mathcal{W}$ is of the form $\mathcal{U} \mathcal{V}$, then the product formula would hold. In the situation of \autoref{thm:symmetry_reduction_combined} we have $\mathcal{U} \cong S_n$ and $\mathcal{V} \cong S_m$, but $\vert \mathcal{W}\vert =\frac{m!}{2} (m!)^{n} n!$, what is strictly larger then $\vert \mathcal{U}\vert \vert \mathcal{V}\vert = n! m!$.
\end{proof}

To conclude the theoretical part of this section, we recall that we defined $\mathcal{W}$ in \autoref{def:strong_joint_symmetry} to be a joint symmetry of the SDP when $\mathcal{W}$ is a group. Part \autoref{thm:symmetry_reduction_combined}~(d) precisely justifies this definition, as it states that whenever $\mathcal{W}$ is a group, we can use it as the group for symmetry reduction, i.e., optimizing in its commutant. Importantly, the strong decomposition theorem (see e.g. \cite[Chap.~2]{Bump2013-ny}) requires this group structure, which confirms that the group case is indeed the relevant one. As already mentioned, $\mathcal{U}$ and $\mathcal{V}$ may be group-theoretic complements of each other. However, this general case is difficult to classify. Therefore, in \autoref{thm:symmetry_reduction_combined}~(a), we provided a general criterion for when $\mathcal{W}$ is a group, and in parts (b) and (c) we presented relevant special cases. To conclude, we recall two interesting examples:\\
 
\begin{enumerate}
    \item In the context of error correction, \autoref{example} highlights the importance of the consideration of the action of the iid-symmetry in AQEC. In the case of \autoref{example}, we are dealing with a dihedral group of order $ 8 $. According to \autoref{thm:symmetry_reduction_combined}, we need to investigate whether a normal subgroup of order $ 2 $ exists or alternatively to disprove the statement, we could use the fact that the cardinality of the generated group is too large. Both questions can be answered affirmatively, as both subgroups $ \mathcal{U} $ and $ \mathcal{V} $ represent reflections within this dihedral group. However, in the dihedral group, only the subgroups of the cyclic group of rotations are normal subgroups. Moreover, the cardinality of $\mathcal{U}\mathcal{V}$ should be $\vert \mathcal{U}\vert \vert \mathcal{V}\vert = 2\cdot 2$, but it is $8$ in fact, contradicting \autoref{thm:symmetry_reduction_combined}~(d) and thus making joint symmetry reduction in the sense of \autoref{def:strong_joint_symmetry} impossible. However, the action discussed in \autoref{sec:commutative_relaxation} can be identified with the case of \autoref{lem:strong_symmetry} (b) and thus a bona-fide joint symmetry reduction in our extended framework is possible as we have done in \autoref{sec:commutative_relaxation}.
    
    \item Assume the following SDP for $\rho_A = \operatorname{tr}_{A_2^n}[\rho_{A_1^n}]$:
\begin{equation}
\begin{aligned}
    \sup & \ \operatorname{tr}[H\rho_A] \\
    \text{s.t.} & \ \operatorname{tr}[\rho_{A} C_j] \leq b_j, \quad 1\leq j \leq m, \\
    &\ H, C_j \ \text{are invariant under a group} \ G, \\
    &\ \rho_{A_1^n} \ \text{is invariant under the natural permutation symmetry} \ S_n.
\end{aligned}
\end{equation}
We can interpret $\operatorname{tr}[H\rho_{A}]$, and similarly $\operatorname{tr}[C_j \rho_A]$, as $\operatorname{tr}[H \otimes \mathds{1}_{A_2^n} \rho_{A_1^n}]$, and thus $H \otimes \mathds{1}_{A_2^n}$ carries a symmetry of the form (with $G_i \cong G$):
\begin{equation}
\begin{aligned}
    G_1 \times \cdots \times G_n &\to \mathcal{U}(\mathcal{H}_{A_1^n}) \\
    (g_1,\ldots,g_n) &\mapsto U_{g_1} \otimes \cdots \otimes U_{g_n}.
\end{aligned}
\end{equation}
Furthermore, $\rho_{A_1^n}$ carries a permutation symmetry, permuting all the $\mathcal{H}_{A_i}$ systems. Since the wreath product is defined as $G \wr S_n \coloneqq G^{n} \rtimes S_n$, this problem admits exactly a semidirect product structure as considered in \autoref{thm:symmetry_reduction_combined}~(b) and \autoref{thm:symmetry_reduction_combined}~(d) implies that an optimal solution already lies in the commutant of $\mathcal{M}_{G \wr S_n}^\prime(\mathcal{H})$.

From this perspective, the difficulty for AQEC is precisely the side information given by Alice’s system, as this structure somehow destroys the usual wreath product structure of the symmetry group generated by the iid-symmetry and the extendibility symmetry as we can observe from this example here where side-information is missing. Actually, we can identify in \autoref{thm:global_symmetry_group} a "hidden" subgroup $A_m$, which arises from the interplay between the group $G$ acting on the side-information and the extendibility symmetry.

\end{enumerate}


\section*{Acknowledgments}

We thank Hoang Ta for discussions about his work \cite{chee2024efficient}. GK sincerely thanks Benjamin Sambale for fruitful discussions. GK thanks Aram Harrow for a fruitful discussion at Beyond IID 2024 in Illinois, USA. GK, JZ, and MB acknowledge support from the Excellence Cluster - Matter and Light for Quantum Computing (ML4Q), and funding by the European Research Council (ERC Grant Agreement No.~948139). OF acknowledges funding by the European Research Council (ERC Grant Agreement No.~851716).


\bibliography{main}
\bibliographystyle{ultimate}

\newpage
\appendix


\section{Symmetries of the Choi matrix}
\label{subsec:symmetries_of_choi}
We dedicate this section to symmetries of $H = C_{\operatorname{id}_L \otimes \mathcal{N}}$. To formalize, we introduce covariant channels as quantum channels preserving symmetry under a unitary transformation group.

\begin{definition}[Covariant Channel]\label{def:covariant_channels}
    A quantum channel $ \mathcal{M}: \mathcal{S}(\mathcal{H}) \to \mathcal{S}(\mathcal{H}) $ is said to be \emph{covariant} with respect to a subgroup $ \mathcal{U} $ of the unitary group $ \mathcal{U}(\mathcal{H}) $ if 
    \begin{align}
    \mathcal{M}(U \cdot U^\dagger) = U \mathcal{M}(\cdot) U^\dagger, \quad \text{for all} \ U \in \mathcal{U}. 
    \end{align}
\end{definition}

Covariant channels are the natural generalization of the invariance of the identity channel with respect to all symmetries. As can easily be checked, the Choi-matrix will then become invariant to the product representation of the group $\mathcal{U}$ from \autoref{def:covariant_channels}.

\begin{proposition}\label{prop:covariant-choi}
Let $\mathcal{H}$ be Hilbert space and $\mathcal{V}$ be a subgroup of the unitary group $\mathcal{U}(\mathcal{H})$ acting on $\mathcal{H}$. 
Suppose $\mathcal{M}: \mathcal{S}(\mathcal{H}) \to \mathcal{S}(\mathcal{H})$ is a channel such that for every $V \in \mathcal{V}$,
\begin{align}
   V \mathcal{M}(\rho) V^\dagger = \mathcal{M}\!(V \rho V^\dagger), 
   \quad \text{for all} \ \rho \in \mathcal{S}(\mathcal{H}).
\end{align}
Then its Choi matrix $C_\mathcal{M}$ is invariant under the unitary action $(V \otimes V)[\cdot](V^\dagger \otimes V^\dagger)$, that is,
\begin{align}
   (V \otimes V)C_\mathcal{M} (V^\dagger \otimes V^\dagger) = C_\mathcal{M}
   \quad\text{for all } V \in \mathcal{V}.
\end{align}
\end{proposition}
\begin{proof}
    This is an easy consequence of the Choi isomorphism.
\end{proof}

\begin{proposition}\label{cor:better_bound_mutual_info}
Let $\rho_{ABC} \in \mathcal{S}(\mathcal{H}_A \otimes \mathcal{H}_B\otimes \mathcal{H}_C)$ be a quantum state such that
    \begin{align}\label{eqn:cond_cor_mutual_info_bound}
        \Tr_A\left[\rho_{ABC}\right] = \rho_B\otimes\rho_C.\,
    \end{align}
    Then, we have
    \begin{align}
        I(AB:C)_{\rho_{ABC}} \leq 2 \log \vert A \vert. 
    \end{align}
\end{proposition}
\begin{proof}
    Write $I(AB:C) = S(AB) + S(C) - S(ABC) \leq S(A) + S(B) + S(C)-S(ABC)$ and apply Araki-Lieb for $S(ABC)$ such that 
    \begin{align}
        S(ABC) \geq \vert S(BC) - S(A)\vert. 
    \end{align}
    Using $S(BC) = S(B) + S(C)$ by assumption yields the desired result. 
\end{proof}

\end{document}